\documentclass[letterpaper]{article} 
\usepackage{aaai2026}  
\usepackage{times}  
\usepackage{helvet}  
\usepackage{courier}  
\usepackage[hyphens]{url}  
\usepackage{graphicx} 
\usepackage{natbib}  
\usepackage{caption} 
\usepackage{algorithm}
\usepackage{algorithmic}

\usepackage{newfloat}
\usepackage{listings}
\DeclareCaptionStyle{ruled}{labelfont=normalfont,labelsep=colon,strut=off} 
\floatstyle{ruled}
\newfloat{listing}{tb}{lst}{}
\floatname{listing}{Listing}
\usepackage{algorithm}
\usepackage{algorithmic}
\usepackage{amsthm}
\usepackage{amsmath}
\usepackage{amssymb}
\usepackage{multirow}
\usepackage{caption}
\usepackage{booktabs} 
\usepackage{siunitx} 
\usepackage{makecell}
\title{Sample-and-Search: An Effective Algorithm for Learning-Augmented  $k$-Median Clustering in High dimensions}
\author{
    Kangke Cheng\textsuperscript{\rm 1}, Shihong Song\textsuperscript{\rm 2}, Guanlin Mo\textsuperscript{\rm 1}, Hu Ding\textsuperscript{\rm 1}\thanks{Corresponding author}
}
\affiliations{
    \textsuperscript{\rm 1}University of Science and Technology of China, Hefei, China\\
    \textsuperscript{\rm 2}School of Informatics, University of Edinburgh, Edinburgh, UK\\
    ke314159@mail.ustc.edu.cn, S.Song-29@sms.ed.ac.uk, moguanlin@mail.ustc.edu.cn, huding@ustc.edu.cn
    
}

\usepackage{cleveref}
\usepackage{tablefootnote}
 \newtheorem{theorem}{Theorem}[section]
 \newtheorem{lemma}[theorem]{Lemma}
  \newtheorem{claim}[theorem]{Claim}
  \newtheorem{proposition}[theorem]{Proposition}

 \newtheorem{definition}{Definition}[section]
 
\begin{document}

\maketitle

\begin{abstract}
In this paper, we investigate the {\em learning-augmented $k$-median clustering} problem, which aims to improve the performance of traditional clustering algorithms by preprocessing the point set with a predictor of error rate $\alpha \in [0,1)$. This preprocessing step assigns potential labels to the points before clustering. We introduce an  algorithm for this problem based on a simple yet effective  sampling method, which substantially improves upon the time complexities of existing algorithms. Moreover, we   mitigate their exponential dependency on the dimensionality of the Euclidean space. 
Lastly, we conduct experiments to compare
 our method with several state-of-the-art learning-augmented $k$-median clustering methods. The experimental results suggest that our proposed approach can significantly reduce the computational complexity in practice, while achieving a lower clustering cost. 
\end{abstract}

\begin{links}
    \link{Code}{https://github.com/KangkeCheng/Learning-Augmented-k-Median-Sample-and-Search}
\end{links}

\section{Introduction}
As the core topics in unsupervised learning, {\em $k$-median} and {\em$k$-means} clusterings are widely applied to numerous  fields, like 
bioinformatics ~\cite{kiselev2017sc3}, 
computer vision ~\cite{caron2020unsupervised}, 
and social network ~\cite{ghaffari2021clustering}. The primary goal of these center-based clustering problems is to partition a set of unlabeled data points into multiple clusters, such that data points within the same cluster are similar to each other (under some  metric),
while data points in different clusters exhibit significant dissimilarity. 
\begin{table*}[ht]
	\centering
        \footnotesize
        \setlength{\tabcolsep}{1mm}

	\begin{tabular}{lccc}
		\toprule
		Methods  & Approximation Ratio &  Label Error Range   & Time Complexity\\
		\midrule
		\makecell[l]{\citet{Ergun2021LearningAugmentedKC}}& $1+O((k\alpha)^{1/4})$&$\tilde{O}(\frac{1}{k})$ & $O(nd\log^3n+\text{poly}(k, \log n))$\\
        \addlinespace
		\citet{Nguyen2022ImprovedLA}& $1+\frac{7\alpha+10\alpha^2-10\alpha^3}{(1-\alpha)(1-2\alpha)}$&$[0, 1/2)$& $O(\frac{1}{1-2\alpha}nd \log^3n\log^2\frac{k}{\delta})$ \\
        \addlinespace
        \citet{huang2025new} &$1+\frac{(6+\epsilon)\alpha-4\alpha^2}{(1-\alpha)(1-2\alpha)} $ &$[0, 1/2)$&$O(nd\log(kd)\log(n\Delta)\cdot(\frac{\sqrt{d}}{\alpha\epsilon})^{O(d)})$\\
        \addlinespace
        Sample-and-Search (ours) & $1+ \frac{(6+\epsilon)\alpha-4\alpha^2}{(1-\alpha)(1-2\alpha)}$ &$[0, 1/2)$& $O(2^{O(1/(\alpha\varepsilon)^4)} nd {\log \frac{k}{\delta}})$\\
		\bottomrule
	\end{tabular}
\caption{A comparison of our Sample-and-Search algorithm with state-of-the-art methods. Here,  the terms $\epsilon > 0$ and $\delta \in (0, 1)$ are the parameters that control the approximation precision and success probability. $\Delta$ denotes the aspect ratio of the given point set.}
\label{t2}
\end{table*}

$k$-means problem seeks to find $k$ centers that minimize the sum of squared Euclidean distances from each point to its nearest center. Formally, the goal is to minimize $\sum_{x \in X} \min_{c \in C} \|x - c\|^2_2$, where
$X$ is the input dataset and $C$ is the set of $k$ centers. Despite its popularity, $k$-means is known to be sensitive to outliers and noise, as the squared distance objective increases the impact of extreme values quadratically. In contrast, the $k$-median problem minimizes the sum of Euclidean distances: $\sum_{x \in X} \min_{c \in C} \|x - c\|_2$, which provides greater robustness to outliers and heavy-tailed distributions. Thus $k$-median clustering is preferable in many practical applications, especially when data is noisy. Therefore, our work focuses on the $k$-median setting, aiming to retain its robustness advantages while addressing the algorithmic challenges through the proposed framework.

\textbf{Learning-Augmented algorithms.} A central challenge in the field of algorithm design lies in simultaneously reducing algorithmic time complexity while maintaining a reliable approximation ratio. The proliferation of large-scale data and the advancement of machine learning bring the opportunity
 to obtain valuable prior knowledge for many classical algorithmic problems. To overcome the often pessimistic bounds of traditional worst-case analysis, the theoretical computer science community has introduced {\em learning-augmented} algorithms~\cite{Hsu2018LearningBasedFE,10.5555/3495724.3496389,10.5555/3540261.3541056,10.1145/3447579,10.1145/3528087}—a new paradigm that falls under the umbrella of ``Beyond Worst-Case Analysis"~\cite{Roughgarden_2021}. 
  The core idea is to design algorithms that can harness auxiliary information, typically from a machine-learned model, to enhance their performance. 

For learning-augmented $k$-means and $k$-median clustering, \citet{pmlr-v178-gamlath22a} explored noisy labels, achieving $(1+O(\epsilon))$-approximation for balanced adversarial noise and $O(1)$-approximation for stochastic noise models. Here, the approximation ratio is a measure of the solution's quality, defined as the ratio of the cost of the algorithm's solution to the cost of the optimal solution. \citet{Ergun2021LearningAugmentedKC} developed a learning-augmented framework where data points are augmented with predicted labels, quantified by an error rate $\alpha \in [0,1)$. They proposed a randomized algorithm for $k$-means problem that achieves a $(1+20\alpha)$-approximation under some specific constraints on $\alpha$ and cluster size in $O(nd\log n)$ time, where $n$ denotes the number of points to be clustered and $d$ denotes the dimension of the space. They also proposed an algorithm for the $k$-median problem that, under the condition that $\alpha=\tilde{O}\left(\frac{1}{k}\right)$, achieves an $\tilde{O}((k\alpha)^{1/4})$-approximation.

\citet{Nguyen2022ImprovedLA} further improved these algorithms. Their $k$-means algorithm directly estimates locally optimal centers dimension-wise across predicted clusters, achieving a $(1 + O(\alpha))$-approximation in $O(nd \log n)$ time when $\alpha \in [0, 1/2)$. Their k-median algorithm significantly improves the approximation guarantee by employing multiple random samplings and pruning techniques. 

More recently, \citet{huang2025new} extended the dimension-wise estimation method for learning-augmented $k$-means to reduce the time complexity while maintaining a similar approximation ratio by using sampling to avoid sorting. They also proposed a $k$-median algorithm that achieves a $1+\frac{(6+\epsilon)\alpha-4\alpha^2}{(1-\alpha)(1-2\alpha)}$-approximation,  which represents the state-of-the-art in terms of approximation ratio for learning-augmented $k$-median clustering, as far as we are aware. However, in their work, the structural differences between $k$-means and $k$-median lead to a fundamental algorithmic gap. For $k$-means, the mean center has a closed-form solution that can be computed independently across dimensions. In contrast, $k$-median centers lack closed-form expressions and cannot be decomposed dimension-wise, making them significantly harder to compute even when point label predictions are available.  As a result, their $k$-median method needs brute-force grid partitioning and searching procedure in the original high-dimentional space, thus introduces an exponential dependence on $d$, which is generally considered unacceptable in practice, particularly in high dimensional scenarios. Hence, a key open problem is: {\em Is it possible to design an algorithm that achieves the state-of-the-art approximation ratio while overcoming the  exponential dependence on the dimension $d$?}
\begin{figure}
    \centering
    \includegraphics[width=1\linewidth]{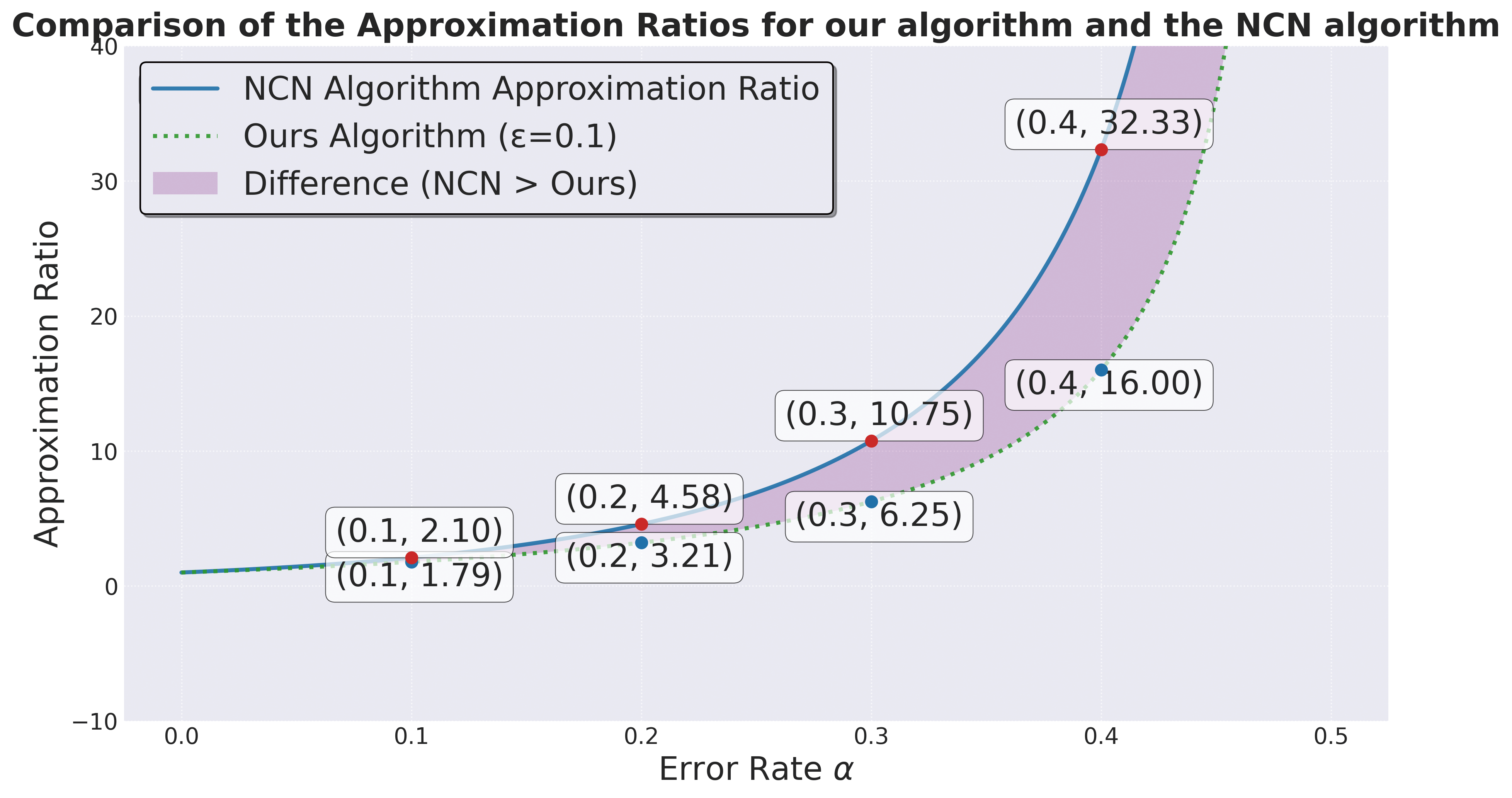}
    \caption{Comparison of the Approximation Ratios for our algorithm (set $\epsilon=0.1$) and the NCN algorithm in term of the change of error rate $\alpha$. This plot shows that our algorithm (green dashed line) consistently achieves a lower approximation ratio than the NCN algorithm~\cite{Nguyen2022ImprovedLA} (blue solid line) across all values of the error rate $\alpha \in [0,1/2)$. The purple shaded area highlights this performance gap, which becomes more pronounced as $\alpha$ increases.}
    \label{approvsalpha}
\end{figure}

\textbf{Our key ideas and main contributions.}
Our key insight is that for each predicted cluster, the true median of the correctly labeled subset lies close to a low-dimensional subspace spanned by a small random sample. This allows us to efficiently discretize the search space using a low-dimensional grid, thus reducing the computational cost.
The main technical challenge is that the predicted clusters may be noisy—some points are misclassified, and the predicted cluster center may be far away from the true one. To tackle this obstacle, we design a novel sampling-and-search framework that can effeictively select appropriate candidate cluster centers.  The key idea is to utilize a greedy search strategy in the aforementioned low-dimensional grid, which neatly avoids to explicitly distinguish between the correclty labeled and mislabeled points.

Our contributions are summarized as follows:
We propose a simple yet effective algorithm for learning-augmented $k$-median clustering. The time complexity is \textbf{linear in $n$ and $d$}, avoiding exponential dependence on the dimension $d$.
At the same time, our algorithm achieves an approximation ratio of $1+ \frac{(6+\epsilon)\alpha-4\alpha^2}{(1-\alpha)(1-2\alpha)}$ for $\alpha <\frac{1}{2}$, matching the state-of-the-art.
Furthermore, we conduct a set of experiments on high-dimensional datasets, demonstrating speedups (up to $10\times$) over prior methods while maintaining relatively high clustering quality. Figure~\ref{approvsalpha} compares the approximation ratio of their algorithm (denoted as NCN) with ours. \Cref{t2} provides a detailed comparison for the results of our and existing methods. 

\subsection{Preliminaries}
\textbf{Notations.} Let $X$ denote the input set of $n$ points  in   $\mathbb{R}^d$. For any two points $p, q \in \mathbb{R}^d$, their Euclidean distance is $\|p-q\|_2$.

Given any point set  $C$,  the distance from a point $p$ to its closest point in $C$ is denoted as $\mathtt{dist}(p, C) = \min_{c \in C} \|p-c\|_2$. In particular, when $C$ is the given set of centers for the point set $X$, the corresponding cost, denoted as $\mathtt{Cost}(X,C)$, is defined as $\mathtt{Cost}(X, C) = \sum_{x \in X} \mathtt{dist}(p, C).\label{for-kmediancost}$ With a slight abuse of notation, we also use $\mathtt{Cost}(P,c)$ to denote the sum of distances from a point set $P$ to a single point $c$.

We denote the optimal $k$ clusters for the given instance $X$ as $\{X_1^*, \ldots, X_k^*\}$, and the set $C^* = \{c_1^*, \ldots, c_k^*\}$ contains their
corresponding optimal centers.

For any point set $P$ and we use $\mathtt{Med}(P)$ to denote its \textbf{median point}, i.e.,
\begin{eqnarray}
\mathtt{Med}(P)=\arg\min_{q\in \mathbb{R}^d}\sum_{p\in P}\|p-q\|_2.
\end{eqnarray}
Therefore, for each $1\leq i\leq k$, $c_i^*=\mathtt{Med}(X_i^*)$. 

\begin{definition}[\textbf{learning-augmented $k$-median clustering}]  
\label{def-la}

Suppose there exists a predictor that outputs a labeled partition $\{\tilde{X}_1, \tilde{X}_2, \dots, \tilde{X}_k\}$ for $X$, parameterized by a label error rate $\alpha\in [0, 1)$, which satisfies:  
\begin{align*}  
|\tilde{X}_i \cap X_i^*| \geqslant (1 - \alpha) \max( |\tilde{X}_i|, |X_i^*|)  
\end{align*}  
where $|\cdot|$ denotes the number of points in a set. The goal of learning-augmented $k$-median clustering is using such a partially correct result to compute 
a center set $C \subset \mathbb{R}^d$ that minimizes $\mathtt{Cost}(P, C)$. 
\end{definition}

We also introduce two important propositions on geometric median point in Euclidean space, which
are essential for our following proofs. 

\begin{figure}
    \centering
    \includegraphics[width=1\linewidth]{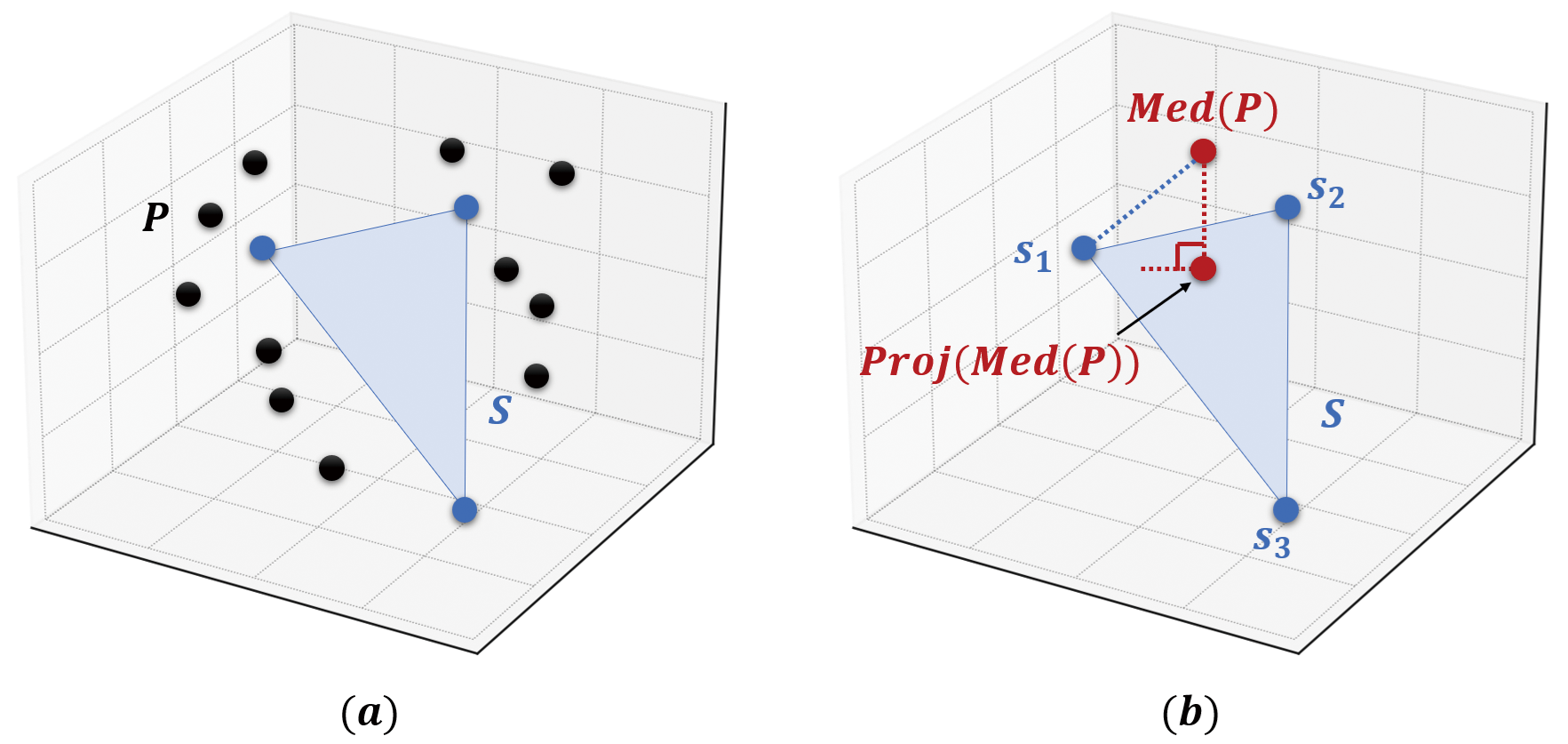} 
    \caption{(a) provides a simplified illustration of how a subspace is generated. We sample a subset $S$ (denoted by the blue points) from the original point set $P$ (denoted by the black points), and $S$ forms a subspace $\mathtt{span}(S)$. (b) shows that $\mathtt{span}(S)$ contains a projection of $\mathtt{Med}(P)$, denoted by $\mathtt{Proj}(\mathtt{Med}(P))$, which is close to  $\mathtt{Med}(P)$. Moreover, $S$ contains a point (\emph{e.g.,} $s_1$) that is within a bounded distance from $\mathtt{Med}(P))$.}
    \label{proposition1}
\end{figure}

\begin{proposition}\cite{10.1145/509907.509947}
\label{kumar1}
Let $P$ be a point set in $\mathbb{R}^d$. Given two parameters $1 > \varepsilon > 0$ and $\gamma > 1$, we draw a random sample $S$ from $P$ of size $\frac{\gamma}{\varepsilon^3} \log \frac{1}{\varepsilon}$. Then, with the probability at least $1-1/\gamma$, the following two events occur:
\textbf{(i)} The flat $\mathtt{span}(S)$ contains a point within a distance of $\frac{\epsilon\cdot\mathtt{Cost}(P,\mathtt{Med}(P))}{|P|}$ from $\mathtt{Med}(P)$, where $\mathtt{span}(S)$ denotes the subspace spanned by the points in $S$.
\textbf{(ii)} The set $S$ contains a point within a distance of $2 \times\frac{\mathtt{Cost}(P, \mathtt{Med}(P))}{|P|}$ from $\mathtt{Med}(P)$.
\end{proposition}
 Proposition~\ref{kumar1} shows that the subspace spanned by a sufficient random sample from $P$ is guaranteed to contain a good approximation of the $\mathtt{Med}(P)$, which enables us to construct a candidate set of centers by partitioning this subspace with a grid to approximate the optimal center. In ~\Cref{proposition1}, we depicts the generation of the subspace from the samples. Proposition~\ref{kumar2} is used to estimate the average cost of the clusters, which in turn guides the design of the grid cell side-length.
\begin{proposition}\cite{10.1145/1667053.1667054} 
\label{kumar2}
Let $P$ be a point set in $\mathbb{R}^d$. Given a parameter $\zeta \in (0, 1/12)$, we randomly sample a point $p_0$ and a set $S$ of size $1/\zeta$ from $P$. Define the value $v = \mathtt{Cost}(S, p_0)$. Then with the probability $>\frac{(1-\zeta^2)^{1/\zeta+1}}{2}$, we have:
$$\frac{v\zeta^3}{2} \leqslant \frac{\mathtt{Cost}(P, \textup{Med}(P))}{|P|} \leqslant \frac{v}{\zeta}. $$
\end{proposition}

\textbf{Upper bound on error rate $\alpha$.} 
We consider the label error rate $\alpha$ to be upper-bounded by $1/2$. As discussed by \citet{Nguyen2022ImprovedLA}, when $\alpha$ reaches $1/2$, the relationship between the predicted and optimal clusters can break down entirely. 

\subsection{Other related work}
\textbf{$k$-median algorithms.} Due to the NP-Hardness of the $k$-median  problem \cite{8948668}, its approximate algorithms have been extensively studied over the past half-century. Although several PTAS (Polynomial-Time Approximation Scheme) algorithms have been proposed, their running time is exponential in either the dimension $d$ or the number of clusters $k$~\cite{10.1145/276698.276718,10.1145/3477541,10.1145/1667053.1667054}, making them impractical for many settings.  
Other algorithms include the $(3+\epsilon)$-approximation algorithm via local search proposed by \citet{10.1145/380752.380755} and the $3.25$-approximation using the LP-rounding approach proposed by \citet{10.1007/978-3-642-31594-7_17}. 

\textbf{$k$-means alogrithms.}
 Similarly with the $k$-median problem, existing PTAS algorithms for the $k$-means problem exhibit an exponential dependence on either $d$ or $k$ ~\cite{10.1145/3477541,10.1145/1667053.1667054}. Simultaneously, the widely used $k$-means++ algorithm has a time complexity of $O(ndk)$ and achieves an approximation ratio of $O(\log k)$ \cite{10.5555/1283383.1283494}, or $O(1)$ for well-separated data \cite{10.1145/2395116.2395117}. 

\section{Our Algorithm And Theoretical Analysis}
In this section, we propose the ``sample-and-search'' algorithm for learning-augmented $k$-median clustering. Our main idea is to 
extract information from predicted labels through uniform sampling, then leverage the properties of median point (based on Proposition~\ref{kumar1}) to construct a candidate center set in a low-dimensional subspace, and finally employ a greedy search approach to find the desired solution from the candidate center set. This sample-and-search strategy avoids 
searching in the original space which may be much larger than the subspace derived from Proposition~\ref{kumar1}, and thereby is able to reduce the total computational complexity to a great extent. In ~\Cref{sec-algandthe}, we introduce the detailed algorithm and our main theoretical result, i.e., \Cref{alg1}. Then, we provide the proof for ~\Cref{alg1} in ~\Cref{sec-proof}.

\subsection{Our Proposed Algorithm And Main Theorem}
\label{sec-algandthe}
We present the Sample-and-Search algorithm in  \Cref{algorithm 1}. In general, the algorithm consists of three main stages:

\begin{enumerate}
    \item \textbf{Sampling-Based Subspace Construction:} For each predicted cluster, we sample a small subset of points to form a \textbf{“basis”} that captures a “neighbor” within a bounded distance from the optimal cluster center. Note that once we have found such a basis, we only need to search within the low-dimensional subspace spanned by this small basis to approximate the optimal center of the cluster. This allows the size of our search space to depend only on $\epsilon$, not on the dimension $d$.
    \item \textbf{Grid-based Candidate Generation:} After generating $k$ appropriate subspaces where each one of them is sufficiently close to the corresponding optimal centers, we construct a grid structure in each of the subspaces to generate $k$ small candidate sets of center points. This eliminates the need to search the original high-dimensional space.
    \item \textbf{Greedy Center Selection:} We select the best center from the candidate set using a cost-minimization greedy selection procedure. 
\end{enumerate}

\begin{algorithm}[htbp]
\caption{\sc{Sample-and-Search for learning-augmented $k$-Median}}
\label{algorithm 1}
\begin{algorithmic}[1]
    \STATE \textbf{Input:}
    A $k$-median instance, consisting of a finite set of points $X \subset \mathbb{R}^d$ and an integer $k$;
    a predicted partition $\{\tilde{X_1}, \dots,\tilde{X_k}\}$ of $X$ with error rate $\alpha \in [0,1/2)$;
    an accuracy parameter $\epsilon \in (0, 1)$ and failure probability $\delta \in (0,1)$.
    \STATE \textbf{Output:}A set $\hat{C} \in {\mathbb{R}^d} $ of centers with $|\hat{C}| = k$ 

    \STATE  Initialize $\hat{C} \leftarrow \emptyset$, $\zeta \leftarrow \frac{1}{13}$
    \FOR{$i \leftarrow 1$ \textbf{to} $k$}
         \STATE  Initialize a candidate set ${C}_i \leftarrow \emptyset$ 
        \FOR{$j \leftarrow 1$ \textbf{to} $\left\lceil\frac{\log(\delta/k)}{\log(0.975)}\right\rceil$}
            \STATE  \textbf{Samplings:} first, randomly sample a point $y^j_i \in \tilde{X_i}$, and then sample two separated sets from $\tilde{X_i}$ uniformly at random: 
            a set $Q^j_i \subseteq \tilde{X_i}$ of size $\lceil\frac{2}{(1-\alpha)\zeta}\rceil$, and a  set ${R}^j_i \subseteq \tilde{X_i}$ of size $\lceil\frac{4\log(1/(\frac{\alpha\epsilon}{2}))}{(1-\alpha)(\frac{\alpha\epsilon}{2})^3}\rceil$
            \FOR{\textbf{each} subset $Q \subseteq {Q}_i^j$ of size $1/\zeta$} 
                \STATE  $v \leftarrow \mathtt{Cost}({Q}, y_i^j)$
                \STATE  $a \leftarrow \frac{v \zeta^3}{2}$ and $b \leftarrow \frac{v}{\zeta}$
                \FOR{\textbf{each integer} $l \in \{\lfloor\log_2 a\rfloor, \dots, \lceil\log_2 b\rceil\}$}
                    \STATE  $t \leftarrow 2^l$
                    \STATE  Run ~\Cref{candidate_set_construction_alg}: ${S} \leftarrow \textsc{CSC}({R}_i^j, t, \alpha, \epsilon)$ 
                    \STATE  ${C}_i \leftarrow {C}_i \cup {S}$
                \ENDFOR
            \ENDFOR
        \ENDFOR
        \STATE  For each $c \in {C}_i$, define ${N}_i(c)$ as the set of $\lceil(1-\alpha)|\tilde{X_i}|\rceil$ points in $\tilde{X_i}$ closest to $c$
        \STATE \textbf{Greedy selection:} find the best candidate for the $i$-th cluster center, $\hat{c}_i \leftarrow \arg\min_{c \in {C}_i} \sum_{x \in {N}_i(c)} \|x-c\|_2$ 
        \STATE $\hat{C} \leftarrow \hat{C} \cup \{\hat{c}_i\}$
    \ENDFOR
    \STATE \textbf{return} $\hat{C}$
\end{algorithmic}
\end{algorithm}
\begin{algorithm}[htbp]
\caption{\sc{Candidate Set Construction (CSC)}}
\label{candidate_set_construction_alg}
\begin{algorithmic}[1]
    \STATE \textbf{Input:}
    A point set ${R} \subset \mathbb{R}^d$;
    an approximate average cost ${t} > 0$;
    parameters {\color{red}}$\alpha $ and $\epsilon$.
    \STATE \textbf{Output:} A set of candidate centers ${S} \subset \mathbb{R}^d$.
    
    \STATE Initialize ${S} \leftarrow \emptyset,\theta \leftarrow \frac{\alpha \epsilon {t}}{4|{R}|}$ 
    \FOR{\textbf{each} point $r \in {R}$}
        \STATE Initialize a candidate set ${S}_r \leftarrow \emptyset$
        \STATE {Construct a grid ${G}_r$ on $\mathtt{span}({R})$ centered at $r$ with side-length $\theta$}
        \STATE {Define a ball ${B}(r, 2{t}) \leftarrow \{x \in \mathbb{R}^d \mid \|x-r\|_2 \le 2{t}\}$}
        \STATE ${S}_r\leftarrow {G}_r\cap {B}(r, 2{t})$ 
        \STATE ${S} \leftarrow {S} \cup {S}_r$
    \ENDFOR
    \STATE \textbf{return} ${S}$
\end{algorithmic}
\end{algorithm}

We present the main theoretical result of our algorithm below. 
\begin{theorem}
\label{alg1}
   Let $1>\epsilon>0$ and $1>\delta>0$ be two parameters. Given an instance $X$ as described in ~\Cref{def-la}, if we assume the error rate  $\alpha < \frac{1}{2}$, then ~\Cref{algorithm 1} can output 
   a solution with the approximation ratio
$\left(1+ \frac{6\alpha-4\alpha^2+\epsilon\alpha}{(1-\alpha)(1-2\alpha)}\right)$
 with probability $1-\delta$. The time complexity is    $O(2^{O(1/(\alpha\varepsilon)^4)} nd {\log \frac{k}{\delta}})$. 
\end{theorem}

\subsection{Proof of ~\Cref{alg1}}
\label{sec-proof}

We divide the proof into three main steps: 
First, we establish that for each predicted cluster, with high probability, the constructed candidate center set contains at least one point that is close to the true median of the correctly labeled subset of the predicted cluster. This is formalized in Lemma~\ref{lmcandidateset}, which leverages the geometric properties from Proposition~\ref{kumar1} and Proposition~\ref{kumar2} under our sampling design.
Second, we analyze the cost of the selected center from the candidate set. In Lemma~\ref{thapproximationratio}, we show that this center yields a clustering cost close to the optimal one, despite the noisy labels, by carefully bounding the additional cost incurred by misclassified points and the optimality of the greedy choice. 

Finally, we aggregate the bounds over all clusters to obtain the total clustering cost, and analyze the size of the candidate set and runtime of our algorithm. 

For convenience, we denote the intersection of $\tilde{X_i}$ and $X_i^*$ as $T_i$, i. e. , $T_i=\tilde{X_i}\cap X_i^*$.
\begin{lemma}
\label{lmcandidateset}
For predicted cluster $\tilde{X_i}$, with a probability of $1 - \frac\delta k$, there exists a point $q \in \tilde{C}_i$ satisfying:
\begin{align}
||q-\mathtt{Med}(T_i)||_2 \leqslant \frac{ \alpha \epsilon\times \mathtt{Cost}(T_i, \mathtt{Med}(T_i))}{ |T_i|
}. 
\end{align}
\end{lemma}
\begin{proof}
   
{\color{red}}First, under the learning-augmented setting, we have $|T_i|\geqslant(1-\alpha)\max(|\tilde{X_i}|, |X_i^*|)>\frac{1}{2}|\tilde{X_i}|$. As we uniformly sample a point $y_i$ from $\tilde{X_i}$, and uniformly sample a set ${Q}_i^j$ from $\tilde{X_i}$ with size $\frac{2}{(1-\alpha)\zeta}$ in the first stage of our algorithm, by employing Markov's inequality, we deduce that, with probability at least $\frac{1}{2}\times\frac{1}{2}=\frac{1}{4}$, the following two events occur simultaneously:\begin{align}
 y_i\in T_i,|T_i\cap {Q}_i^j|\geqslant \frac{1}{\zeta}\nonumber.
\end{align}
Now assume both of these events occur. There exists a subset
${Q} \subseteq (T_i\cap {Q}_i^j)$ of size $\frac{1}{\zeta}$. According to  Proposition~\ref{kumar2}, if we set $p_0=y_i^j,S=Q,P=T_i$, for 
$v=\mathtt{Cost}({Q}, y_i^j)$, we have
\begin{align}
\frac{v\zeta^3}{2}\leqslant \frac{\mathtt{Cost}(T_i, \mathtt{Med}(T_i))}{|T_i|}\leqslant \frac{v}{\zeta}
\end{align} with probability at least $\frac{(1-\zeta^2)^{1/\zeta+1}}{2}$. 
In the second stage of our algorithm, we iterate over all subsets of ${Q}_i^j$ of size $\frac{1}{\zeta}$, therefore, there exists an integer $l$ in the interval $\lfloor \log{\frac{v\zeta^3}{2}}, \log{\frac{v}{\zeta}}\rceil$ such that $t=2^l$ satisfies
\begin{align}
    t/2\leqslant\frac{\mathtt{Cost}(T_i, \mathtt{Med}(T_i))}{|T_i|}\leqslant t. \label{eq4}
\end{align}
Similarly, as we uniformly sample a set ${R}_i^j$ from $\tilde{X_i}$
with size $\frac{4\log(1/\frac{\alpha\epsilon}{2})}{(1-\alpha)(\frac{\alpha\epsilon}{2})^{3}}$ in the first stage of our algorithm, by employing Markov's inequality, we have
$|T_i\cap {R}_i^j|\geqslant \frac{2\log(1/\frac{\alpha\epsilon}{2})}{(\frac{\alpha\epsilon}{2})^{3}} $
with probability at least $1/2$. 
Thus, according to the  Proposition~\ref{kumar1}, with
probability at least $1/2$ , the following two events happen:
\begin{enumerate}
    \item  The flat $\mathtt{span}({R}_i^j)$ contains a point at a distance $\leqslant \frac{\alpha\epsilon\mathtt{Cost}(T_i,\mathtt{Med}(T_i))}{2|T_i|}$ from $\mathtt{Med}(T_i)$,
    \item ${R}_i^j$ contains a point at a distance  $\leqslant 2 \times \frac{\mathtt{Cost}(T_i, \mathtt{Med}(T_i))}{|T_i|}$ from the center $\mathtt{Med}(T_i)$. 
\end{enumerate}
So, the flat $\mathtt{span}({R}_i^j)$ contains a point $o$ such that
\begin{align}
||o-\mathtt{Med}(T_i)||_2\leqslant\frac{\alpha\epsilon\times \mathtt{Cost}(T_i, \mathtt{Med}(T_i))}{2|T_i|}. \label{eq5}\end{align}
Therefore, under the construction of the grid in  \Cref{candidate_set_construction_alg}, there must exist a point \( q \in {S} \) satisfying  
\begin{align}
    ||q-o||_2 \leqslant \frac{\alpha\varepsilon t}{4} \leqslant  \frac{\alpha\varepsilon\times\mathtt{Cost}(T_i, \mathtt{Med}(T_i))}{2|T_i|}, 
    \label{eq6}
\end{align}
where the second inequality comes from inequality \eqref{eq4}.
Combining inequality \eqref{eq5} and inequality \eqref{eq6}, by triangle inequality, we have
\begin{align}
||q-\mathtt{Med}(T_i)||_2\leqslant\frac{\alpha\epsilon \times\mathtt{Cost}(T_i, \mathtt{Med}(T_i))}{|T_i|}\nonumber. 
\end{align}
Now, we calculate the probability that all events succeed in a single trial. The combined success probability is $\frac{(1-\zeta^2)^{1/\zeta+1}}{32}.$
We have $\frac{(1-\zeta^2)^{1/\zeta+1}}{32}>\frac{e^{-\frac{\zeta}{1-\zeta}}}{32}\geq 0.025$
when $\zeta \in (0,1/12)$. Here, the first inequality is a direct application of $\ln(1-x)>-\frac{x}{1-x}$ and the second inequality is obtained by leveraging the fact that the function is monotonically decreasing.
Therefore, the probability of failure for each trial is less than $0.975$. Since we perform $\left\lceil\frac{\log(\delta/k)}{\log(0.975)}\right\rceil$runs, the overall success probability is therefore greater than $ 1-0.975^{\frac{\log(\delta/k)}{\log(0.975)}}=1-\frac{\delta}{k}.$
\end{proof}

We now turn to evaluate the clustering cost incurred by the selected centers. Lemma~\ref{thapproximationratio} plays a central role in this analysis—it quantifies how far the selected center might be from the true cluster center due to noisy labels and sampling variance, and how this error translates into overall clustering cost.
\begin{lemma}
\label{thapproximationratio}
For each predicted cluster $\tilde{X_i}$, we have:
\begin{align}
  \mathtt{Cost}(X_i^*, \hat{c_i})\leqslant\left(1+\frac{6\alpha-4\alpha^2+\alpha\epsilon}{(1-\alpha)(1-2\alpha)}\right)\mathtt{Cost}(X_i^*, c_i^*). \nonumber
\end{align}
\end{lemma}

\begin{proof}
A critical component of the analysis is to relate the selected center $\hat{c}_i$to the true center $c_i^*$ and the median of correctly labeled points $\mathtt{Med}(T_i)$. 
We begin by splitting the cost into two parts, 
\begin{align}
    \mathtt{Cost}(X_i^*, \hat{c}_i)=\mathtt{Cost}(T_i, \hat{c}_i)+\mathtt{Cost}(X_i^*\backslash T_i, \hat{c}_i) \label{for-split}
\end{align}
We focus on the   term ``$\mathtt{Cost}(T_i, \hat{c}_i)$'' first. 
 To establish the equality, we first compute the additional cost incurred by assigning points in $T_i$ to $\hat{c}_i$ by decomposing the sets $T_i$ and $N_i(\hat{c}_i)$ into three disjoint partitions:
$S_1 = T_i \cap N_i(\hat{c}_i),
S_2 = T_i \setminus N_i(\hat{c}_i),
S_3 = N_i(\hat{c}_i) \setminus T_i.$
This implies $T_i = S_1 \cup S_2$ and $N_i(\hat{c}_i) = S_1 \cup S_3$. Then the $\mathtt{Cost}(T_i, \hat{c}_i)-\mathtt{Cost}(T_i, \mathtt{Med}(T_i))$ can be written as
\begin{align}
& \mathtt{Cost}(T_i, \hat{c}_i) - \mathtt{Cost}(T_i, \mathtt{Med}(T_i)) \nonumber\\
& = \left[\mathtt{Cost}(S_1, \hat{c}_i) + \mathtt{Cost}(S_2, \hat{c}_i)\right]  \nonumber\\
&\quad- \left[\mathtt{Cost}(S_1, \mathtt{Med}(T_i)) + \mathtt{Cost}(S_2, \mathtt{Med}(T_i))\right] \nonumber\\
& = \left[\mathtt{Cost}(S_1, \hat{c}_i) + \mathtt{Cost}(S_3, \hat{c}_i)\right]  \nonumber\\
&\quad- \left[\mathtt{Cost}(S_1, \mathtt{Med}(T_i)) + \mathtt{Cost}(S_3, \mathtt{Med}(T_i))\right] \nonumber\\
&\quad- \left[\mathtt{Cost}(S_3, \hat{c}_i) - \mathtt{Cost}(S_3, \mathtt{Med}(T_i))\right] \nonumber\\
&\quad+ \left[\mathtt{Cost}(S_2, \hat{c}_i) - \mathtt{Cost}(S_2, \mathtt{Med}(T_i))\right] \nonumber\\
&= \left[\mathtt{Cost}(N_i(\hat{c}_i), \hat{c}_i) - \mathtt{Cost}(N_i(\hat{c}_i), \mathtt{Med}(T_i))\right] \nonumber\\
&\quad- \left[\mathtt{Cost}(N_i(\hat{c}_i)\setminus T_i, \hat{c}_i) - \mathtt{Cost}(N_i(\hat{c}_i)\setminus T_i, \mathtt{Med}(T_i))\right] \nonumber\\
&\quad+\left[\mathtt{Cost}(T_i\setminus N_i(\hat{c}_i), \hat{c}_i) - \mathtt{Cost}(T_i\setminus N_i(\hat{c}_i), \mathtt{Med}(T_i))\right]\label{eq:short_proof}
\end{align}

We also have $|T_i\setminus {N}_i(\hat{c}_i))|\leqslant\alpha |\tilde{X_i}|$ and $ |{N}_i(\hat{c}_i))\setminus T_i|\leqslant\alpha |\tilde{X_i}|$.
So we can find an upper bound for $\mathtt{Cost}(T_i\setminus N_i(\hat{c}_i), \hat{c}_i) - \mathtt{Cost}(T_i\setminus N_i(\hat{c}_i), \mathtt{Med}(T_i))$ by triangle inequality as
\begin{align}
&\mathtt{Cost}(T_i\setminus N_i(\hat{c}_i), \hat{c}_i) - \mathtt{Cost}(T_i\setminus N_i(\hat{c}_i), \mathtt{Med}(T_i))\nonumber\\&\leqslant |T_i\setminus N_i|\times||\mathtt{Med}(T_i)-\hat{c}_i||_2\leqslant\alpha |\tilde{X_i}|||\mathtt{Med}(T_i)-\hat{c}_i||_2\label{eq11}
.\end{align}
We can obtain inequality $\|\mathtt{Med}(T_i)-\hat{c}_i||_2\leqslant\frac{(2+\alpha\epsilon)\mathtt{Cost}(X_i^*, c_i^*)}{(1-2\alpha)|\tilde{\tilde{X_i}}|}$ through triangle inequality (the detailed derivation is provided in the full version of the paper.).
 Then we have inequality \eqref{eq11}
$\leqslant\frac{(2\alpha+\alpha^2\epsilon)\mathtt{Cost}(X_i^*, c_i^*)}{1-2\alpha}.$

Similarly, we obtain \begin{align}&\mathtt{Cost}(N_i(\hat{c}_i)\setminus T_i, \hat{c}_i) - \mathtt{Cost}(N_i(\hat{c}_i)\setminus T_i, \mathtt{Med}(T_i))
\nonumber\\&\leqslant\frac{(2\alpha+\alpha^2\epsilon)\mathtt{Cost}(X_i^*, c_i^*)}{1-2\alpha}.\label{eq10}\end{align}
Now we find a upper bound for $\mathtt{Cost}(N_i(\hat{c}_i), \hat{c}_i) - \mathtt{Cost}(N_i(\hat{c}_i), \mathtt{Med}(T_i)).$ 
Because of our greedy selection, we have $\mathtt{Cost}(N_i(\hat{c}_i),\hat{c}_i)\leqslant\mathtt{Cost}(N_i(q),q)\leqslant\mathtt{Cost}(N_i(\hat{c}_i),q)$. So
\begin{align}
    &\mathtt{Cost}(N_i(\hat{c}_i), \hat{c}_i) - \mathtt{Cost}(N_i(\hat{c}_i), \mathtt{Med}(T_i))\nonumber\\&\leqslant\mathtt{Cost}(N_i(\hat{c}_i),q)-\mathtt{Cost}(N_i(\hat{c}_i), \mathtt{Med}(T_i))\nonumber\\&\leqslant|(1-\alpha)|\tilde{X_i}||\times||q-\mathtt{Med}(T_i)||_2.\label{eq14}
    \end{align}
By applying Lemma~\ref{lmcandidateset}, we can obtain inequality \eqref{eq14} $\leqslant\alpha \epsilon\times \mathtt{Cost}(T_i, \mathtt{Med}(T_i)).$
Putting inequality \eqref{eq11}, inequality \eqref{eq10} and inequality \eqref{eq14} together, we obtain the following bound for the left side of equation \eqref{for-split}
\begin{align}
&\mathtt{Cost}(T_i, \hat{c}_i) - \mathtt{Cost}(T_i, c_i^*)\nonumber\\&\leqslant
    \mathtt{Cost}(T_i, \hat{c}_i) - \mathtt{Cost}(T_i, \mathtt{Med}(T_i))\nonumber\\&\leqslant \frac{(4\alpha+\epsilon\alpha)\mathtt{Cost}(X_i^*, c_i^*)}{1-2\alpha}\label{eq12}.
\end{align}

Next, we consider the second term ``$\mathtt{Cost}(X_i^*\backslash T_i, \hat{c}_i)$'' in equation \eqref{for-split}. By triangle inequality 
\begin{align}
\mathtt{Cost}(X_i^*\backslash T_i, \hat{c}_i)&\leqslant \mathtt{Cost}(X_i^*/T_i, c_i^*) \nonumber\\&+|X_i^*\backslash T_i|\times||\hat{c}_i-c_i^*||_2, \label{eq36}
\end{align}
Subsequently, we bound $|X_i^*\backslash T_i|$. It follows from \Cref{def-la} that
$(1-\alpha)|X_i^*|\leqslant|T_i|,|T_i|\leqslant |\tilde{X_i}|, $
so, we can bound $|X_i^*\backslash T_i|$ as 
\begin{align}|X_i^*\backslash T_i|=|X_i^*|-|T_i|\leqslant\alpha|X_i^*| \leqslant\frac{\alpha |\tilde{X_i}|}{1-\alpha}. \label{eq37}\end{align}
Combining \eqref{eq36}, \eqref{eq37}, and the inequality $\|\hat{c}_i-c_i^*||_2\leqslant\frac{(2+\alpha\epsilon)\mathtt{Cost}(X_i^*, c_i^*)}{(1-2\alpha)|\tilde{{X_i}}|}$ we obtain in the full version of the paper, we have
\begin{align}
    &\mathtt{Cost}(X_i^*\backslash T_i, \hat{c}_i)-\mathtt{Cost}(X_i^*\backslash T_i, c_i^*)\nonumber\\&\leqslant \frac{\alpha(2+\alpha\epsilon)\mathtt{Cost}(X_i^*, c_i^*)}{(1-\alpha)(1-2\alpha)}.  \label{eq15}
\end{align}

Now, we derive the final approximation guarantee. Combining inequality \eqref{eq12}
and inequality \eqref{eq15}, we have
\begin{align}
\mathtt{Cost}(X_i^*, \hat{c}_i)&=\mathtt{Cost}(T_i, \hat{c}_i)+\mathtt{Cost}(X_i^*/T_i, \hat{c}_i) \nonumber\\
	&\leqslant 	\left(1+ \frac{6\alpha-4\alpha^2+\epsilon\alpha}{(1-\alpha)(1-2\alpha)}\right)\mathtt{Cost}(X_i^*, c_i^*).\nonumber
\end{align}

\end{proof}

We now proceed to formally prove Theorem 2.1 by establishing both the approximation ratio and the runtime complexity.

\begin{proof}[\textbf{Proof of \Cref{alg1}}]
Our first step is to compute the approximation ratio of the algorithm. In each cluster, by Lemma~\ref{thapproximationratio}, we obtain 
\begin{align}
    \mathtt{Cost}(X_i^*, \hat{c}_i)
	\leqslant \left(1+ \frac{6\alpha-4\alpha^2+\epsilon\alpha}{(1-\alpha)(1-2\alpha)}\right)\mathtt{Cost}(X_i^*, c_i^*). \nonumber
\end{align}
Therefore, for the entire instance, we have
\begin{align}&\sum_{i\in[k]}\mathtt{Cost}(X_{i}^{*}, \{\widehat{c}_{j}\}_{j=1}^{k})\nonumber	\\&\leqslant \left(1+ \frac{6\alpha-4\alpha^2+\epsilon\alpha}{(1-\alpha)(1-2\alpha)}\right)\sum_{i\in[k]}\mathtt{Cost}(X_i^*, c_i^*). \nonumber
\end{align}
We now assess the time complexity of the algorithm. This involves analyzing the size of the candidate center set generated via sampling and grid discretization, and the cost incurred in evaluating all candidate centers.

First, we compute the size of set of candidate centers. The size of the candidate center set we ultimately construct is
\begin{align}
&O\left(|R|\left(\frac{1}{(\alpha\epsilon)^4}\log\frac{1}{(\alpha\epsilon)}\right)^{O(|R|)}\log\frac{k}{\delta}\right)\nonumber\\&= O\left(2^{O(\frac{1}{(\alpha\epsilon)^3}\log^2\frac{1}{(\alpha\epsilon)})}\log\frac{k}{\delta}\right)\leqslant O\left(2^{O(\frac{1}{(\alpha\epsilon)^4})}\log\frac{k}{\delta}\right). \nonumber
\end{align}
For each candidate point within the candidate center set, the time needed to calculate its cost is $|\tilde{X_i}|d$. Consequently, the overall time complexity of the algorithm is

\begin{align}
\sum_{i\in[k]}2^{O(\frac{1}{(\alpha\epsilon)^4})}|\tilde{X_i}|d\log\frac{k}{\delta}=2^{O(\frac{1}{(\alpha\epsilon)^4})} nd\log \frac{k}{\delta}. \nonumber
\end{align}\nonumber
Next, we analyze the success probability of the algorithm. 
As we obtained in Lemma \ref{lmcandidateset}, the success probability in each cluster of the algorithm is $1-\frac{\delta}{k}$, therefore, by the union bound, the overall success probability of the algorithm $\geqslant 1-k\times\frac{\delta}{k}=1-\delta$. 
\end{proof}

\section{Experiment}\label{Experiment}
We evaluated our algorithms on real-world datasets. The experiments were conducted on a server with an Intel(R) Xeon(R) Gold 6154 CPU and 1024GB of RAM. For all experiments, we report the average clustering cost and its standard deviation over 10 independent runs.

\textbf{Datasets}. Following the work of \citet{Nguyen2022ImprovedLA}, \citet{Ergun2021LearningAugmentedKC} and \citet{huang2025new}, we evaluate our algorithms on the CIFAR-10 $(n=50,000,d=3,072)$ \cite{krizhevsky2009learning}, PHY $(n=10,000,d=50)$ \cite{kddcup2004}, and MNIST  $(n=1,797,d=64)$ \cite{6296535} datasets using a range of error rates $\alpha$. We additionally evaluated our algorithm's performance on another high dimensional dataset Fashion-MNIST $(n=60000, d=784)$ \cite{fashion_mnist} 

\textbf{Predictor Generation and Error Simulation} To evaluate our algorithms, we first computed a ground-truth partition for each dataset using Lloyd's algorithm initialized with KMedoids++(denoteed as KMed++). We then generated corrupted partitions with the error rate, $\alpha$, by randomly selecting an $\alpha$ fraction of points in each true cluster and reassigning them to randomly chosen cluster(denoted as Predictor). To ensure a fair comparison, every algorithm was tested on the exact same set of corrupted labels for any given error rate $\alpha$.

\textbf{Algorithms} In our experiments, we evaluate our proposed Sample-and-Search algorithm . We compare its performance against other state-of-the-art learning-augmented methods, including the algorithm from \citet{Ergun2021LearningAugmentedKC} (denoted as EFS+), \citet{Nguyen2022ImprovedLA} (denoted as NCN) and the recent work by \citet{huang2025new} (denoted as HFH+). As noted by ~\citet{Nguyen2022ImprovedLA}, the true error rate $\alpha$ is generally unknown in practice, which necessitates a search for its optimal value. To ensure a fair comparison, we implement a uniform hyperparameter tuning strategy for all evaluated algorithms . Specifically, we iterate over 10 candidate values for $\alpha$, which are chosen from uniformly spaced points in the interval $[0.01, 0.5]$. For each method, the $\alpha$ that minimizes the resulting $k$-median clustering cost is chosen to produce the final output. To assess the final clustering quality against the ground-truth labels, we additionally report the {\em Adjusted Rand Index (ARI)} and {\em Normalized Mutual Information (NMI)} in the full version of the paper..

\textbf{Results}
We present a comparative evaluation of our algorithm against several baselines in \Cref{tab:fashion_mnist_varied_alpha}  and \Cref{tab:phy_varied_k_compressed}.
\Cref{tab:fashion_mnist_varied_alpha} details the performance on the Fashion-MNIST ($n=60000, d=784$) for a fixed $k=10$ across a range of $\alpha$ values. \Cref{tab:phy_varied_k_compressed} shows the results on the PHY $(n=10,000,d=50)$ with a fixed $\alpha=0.2$ for various choices of $k$. Both sets of results demonstrate that our algorithm is substantially faster than all competing methods while generally achieving better approximation quality. Additional experiments on other datasets and a more detailed presentation of the results are available in the supplementary material. On these datasets, our algorithm also demonstrates significant advantages in terms of both running time and cost.
\section{Conclusion and Future work}
In this paper, we study the learning-augmented $k$-median clustering problem. We first introduce an algorithm for this problem based on a simple yet effective sampling method, then study its quality guarantees in theory, and finally conduct a set of experiments to compare  with other learning-augmented $k$-median algorithms. Both theoretical and experimental results demonstrate that our method achieves the state-of-the-art approximation ratio with higher efficiency than existing methods. Following this work, there are several opportunities to further improve our methods from both theoretical and practical perspectives. For example, is it possible to further reduce the time complexity of the algorithm by mitigating or eliminating the exponential dependence on $\epsilon$? Can the approximation ratio be further improved without a significant increase in time complexity? Furthermore, could a learning-augmented clustering algorithm be designed for the streaming model to more effectively handle large-scale data? 

\begin{table}[htbp]
\centering

\footnotesize
\setlength{\tabcolsep}{1mm}
\begin{tabular}{
    l
    l
    c c
    c c
}
\hline
\multicolumn{2}{c}{Condition} & \multicolumn{2}{c}{Cost} & \multicolumn{2}{c}{Time(s)} \\
\cmidrule(lr){3-4} \cmidrule(lr){5-6}
{$\alpha$} & {Method} & {Avg.} & {Std. Dev.} & {Avg.} & {Std. Dev.} \\
\midrule
\multicolumn{2}{l}{$0$}{KMed++} & 8.4054e+07 & - & - & - \\
\hline

\multirow{5}{*}{0.1}
& Predictor & 8.4259e+07 & - & - & - \\
& EFS+      & 8.4050e+07 & 115.17 & 270.47 & 12.97 \\
& HFH+    & 8.4049e+07 & 834.92 & 749.72 & 18.47 \\
& NCN      & 8.4050e+07 & 181.80 & 272.22 & 4.37  \\
& Ours     & \textbf{8.4048e+07} & 933.64 & \textbf{47.37}  & 0.78  \\
\hline

\multirow{5}{*}{0.2}
& Predictor & 8.4935e+07 & - & - & - \\
& EFS+     & 8.4057e+07 & 287.83 & 283.13 & 24.07 \\
& HFH+    & 8.4053e+07 & 1598.91 & 751.66 & 25.42 \\
& NCN      & 8.4057e+07 & 309.52 & 282.97 & 13.78 \\
& Ours     & \textbf{8.4052e+07} & 961.03 & \textbf{47.96}  & 3.33  \\
\hline

\multirow{5}{*}{0.3}
& Predictor & 8.6223e+07 & - & - & - \\
& EFS+      & 8.4076e+07 & 467.75 & 282.57 & 8.66  \\
& HFH+    & 8.4065e+07 & 3527.00 & 751.13 & 22.67 \\
& NCN      & 8.4077e+07 & 695.49 & 299.50 & 22.54 \\
& Ours     & \textbf{8.4062e+07} & 3848.20 & \textbf{45.38}  & 1.33  \\
\hline

\multirow{5}{*}{0.4}
& Predictor & 8.8209e+07 & - & - & - \\
& EFS+      & 8.4109e+07 & 631.67 & 297.27 & 13.66 \\
& HFH+    & 8.4101e+07 & 11512.25 & 758.16 & 29.89 \\
& NCN      & 8.4111e+07 & 1206.20 & 302.45 & 25.71 \\
& Ours     & \textbf{8.4100e+07} & 12684.42 & \textbf{45.29}  & 2.18  \\
\hline

\multirow{5}{*}{0.5}
& Predictor & 9.0897e+07 & - & - & - \\
& EFS+      & 8.4150e+07 & 1320.47 & 304.67 & 11.82 \\
& HFH+    & 8.4148e+07 & 12671.01 & 751.08 & 26.53 \\
& NCN      & 8.4152e+07 & 2503.70 & 305.95 & 10.98 \\
& Ours     & \textbf{8.4145e+07} & 17385.62 & \textbf{47.87}  & 2.04  \\
\hline
\end{tabular}
\caption{Performance comparison on Fashion-MNIST dataset with $k=10$ and varied $\alpha$.}
\label{tab:fashion_mnist_varied_alpha}
\end{table}

\begin{table}[H]
\centering
\footnotesize
\setlength{\tabcolsep}{1mm}
\begin{tabular}{
    l
    l
    l l
    l l
}
\hline
\multicolumn{2}{c}{Condition} & \multicolumn{2}{c}{Cost} & \multicolumn{2}{c}{Time(s)} \\
\cmidrule(lr){3-4} \cmidrule(lr){5-6}
{k} & {Method} & {Avg.} & {Std. Dev.} & {Avg.} & {Std. Dev.} \\
\hline

\multirow{6}{*}{10}
& KMed++  & 2.0224e+08 & - & - & - \\
& Predictor & 2.0427e+08 & - & - & - \\
& EFS+      & 2.0204e+08 & 4444.80 & 362.31 & 52.30 \\
& HFH+    & 2.0147e+08 & 105661.96 & 42.33  & 1.91 \\
& NCN       & 2.0163e+08 & 109131.95 & 160.15 & 14.33 \\
& Ours     & \textbf{2.0134e+08} & 82812.18 & \textbf{20.72}  & 0.59 \\
\hline

\multirow{6}{*}{30}
& KMed++  & 8.4404e+07 & - & - & - \\
& Predictor & 8.5018e+07 & - & - & - \\
& EFS+      & 8.4490e+07 & 721.59 & 294.21 & 54.13 \\
& HFH+    & 8.4404e+07 & 4372.98 & 42.26  & 2.05 \\
& NCN       & 8.4480e+07 & 14266.80 & 221.81 & 49.30 \\
& Ours     & \textbf{8.4404e+07} & 3043.38 & \textbf{27.08}  & 3.26 \\
\hline
\multirow{6}{*}{50}
& KMed++   & 6.2758e+07 & - & - & - \\
& Predictor & 6.3111e+07 & - & - & - \\
& EFS+     & 6.2796e+07 & 503.13 & 285.51 & 22.48 \\
& HFH+    & 6.2755e+07 & 1072.71 & 44.69  & 1.34 \\
& NCN     & 6.2791e+07 & 5662.71 & 208.89 & 29.26 \\
& Ours     & \textbf{6.275456e+07} & 677.03 & \textbf{36.87}  & 0.73 \\
\hline
\end{tabular}
\caption{Performance comparison on PHY dataset with fixed $\alpha=0.2$ and varied $k$.}
\label{tab:phy_varied_k_compressed}
\end{table}
\section{Acknowledgments}
This work was supported in part by  the NSFC through grants No. 62432016 and No. 62272432,  the National Key R\&D program of China through grant 2021YFA1000900, and the Provincial
NSF of Anhui through grant 2208085MF163. The authors would also like to thank the anonymous
reviewers for their valuable comments and suggestions.
\bibliography{aaai2026}

@inproceedings{10.1145/380752.380755,
author = {Arya, Vijay and Garg, Naveen and Khandekar, Rohit and Meyerson, Adam and Munagala, Kamesh and Pandit, Vinayaka},
title = {Local search heuristic for k-median and facility location problems},
year = {2001},
isbn = {1581133499},
publisher = {Association for Computing Machinery},
address = {New York, NY, USA},
url = {https://doi.org/10.1145/380752.380755},
doi = {10.1145/380752.380755},
booktitle = {Proceedings of the Thirty-Third Annual ACM Symposium on Theory of Computing},
pages = {21–29},
numpages = {9},
location = {Hersonissos, Greece},
series = {STOC '01}
}

@inproceedings{10.1007/978-3-642-31594-7_17,
author = {Charikar, Moses and Li, Shi},
title = {A dependent LP-rounding approach for the k-median problem},
year = {2012},
isbn = {9783642315930},
publisher = {Springer-Verlag},
address = {Berlin, Heidelberg},
url = {https://doi.org/10.1007/978-3-642-31594-7_17},
doi = {10.1007/978-3-642-31594-7_17},
booktitle = {Proceedings of the 39th International Colloquium Conference on Automata, Languages, and Programming - Volume Part I},
pages = {194–205},
numpages = {12},
location = {Warwick, UK},
series = {ICALP'12}
}

@inproceedings{
Ergun2021LearningAugmentedKC,
title={Learning-Augmented \$k\$-means Clustering},
author={Jon C. Ergun and Zhili Feng and Sandeep Silwal and David Woodruff and Samson Zhou},
booktitle={International Conference on Learning Representations},
year={2022},

}

@article{kiselev2017sc3,
  title={SC3: consensus clustering of single-cell RNA-seq data},
  author={Kiselev, Vladimir Y and Kirschner, Kristina and Schaub, Michael T and Andrews, Tallulah and Yiu, Andrew and Chandra, Tamir and Natarajan, Kedar N and Reik, Wolf and Barahona, Mauricio and Green, Anthony R and Hemberg, Martin},
  journal={Nature methods},
  volume={14},
  number={5},
  pages={483--486},
  year={2017},
  publisher={Nature Publishing Group}
}

@inproceedings{caron2020unsupervised,
  title={Unsupervised learning of visual features by swapping assignments between views},
  author={Caron, Mathilde and Misra, Ishan and Mairal, Julien and Goyal, Priya and Bojanowski, Piotr and Joulin, Armand},
  booktitle={Advances in Neural Information Processing Systems},
  volume={33},
  pages={9912--9924},
  year={2020}
}

@inproceedings{ghaffari2021clustering,
  title={Clustering and high-dimensional representation of social network users' behavior for bot detection},
  author={Ghaffari, Meysam and Mosavi, Ahmadreza and Shamshirband, Shahab},
  booktitle={Companion Proceedings of the Web Conference 2021},
  pages={19--22},
  year={2021}
}

@INPROCEEDINGS{8948668,
  author={Cohen-Addad, Vincent and C.S.Karthik},
  booktitle={2019 IEEE 60th Annual Symposium on Foundations of Computer Science (FOCS)}, 
  title={Inapproximability of Clustering in Lp Metrics}, 
  year={2019},
  volume={},
  number={},
  pages={519-539},
  doi={10.1109/FOCS.2019.00040}}

@inproceedings{10.1145/276698.276718,
author = {Arora, Sanjeev and Raghavan, Prabhakar and Rao, Satish},
title = {Approximation schemes for Euclidean k-medians and related problems},
year = {1998},
isbn = {0897919629},
publisher = {Association for Computing Machinery},
address = {New York, NY, USA},
url = {https://doi.org/10.1145/276698.276718},
doi = {10.1145/276698.276718},
booktitle = {Proceedings of the Thirtieth Annual ACM Symposium on Theory of Computing},
pages = {106–113},
numpages = {8},
location = {Dallas, Texas, USA},
series = {STOC '98}
}

@inproceedings{10.1145/509907.509947,
author = {Badoiu, Mihai and Har-Peled, Sariel and Indyk, Piotr},
title = {Approximate clustering via core-sets},
year = {2002},
isbn = {1581134959},
publisher = {Association for Computing Machinery},
address = {New York, NY, USA},
url = {https://doi.org/10.1145/509907.509947},
doi = {10.1145/509907.509947},
booktitle = {Proceedings of the Thiry-Fourth Annual ACM Symposium on Theory of Computing},
pages = {250–257},
numpages = {8},
location = {Montreal, Quebec, Canada},
series = {STOC '02}
}

@article{10.1145/3477541,
author = {Cohen-Addad, Vincent and Feldmann, Andreas Emil and Saulpic, David},
title = {Near-linear Time Approximation Schemes for Clustering in Doubling Metrics},
year = {2021},
issue_date = {December 2021},
publisher = {Association for Computing Machinery},
address = {New York, NY, USA},
volume = {68},
number = {6},
issn = {0004-5411},
url = {https://doi.org/10.1145/3477541},
doi = {10.1145/3477541},
journal = {J. ACM},
month = oct,
articleno = {44},
numpages = {34}
}

@article{10.1145/1667053.1667054,
author = {Kumar, Amit and Sabharwal, Yogish and Sen, Sandeep},
title = {Linear-time approximation schemes for clustering problems in any dimensions},
year = {2010},
issue_date = {January 2010},
publisher = {Association for Computing Machinery},
address = {New York, NY, USA},
volume = {57},
number = {2},
issn = {0004-5411},
url = {https://doi.org/10.1145/1667053.1667054},
doi = {10.1145/1667053.1667054},
journal = {J. ACM},
month = feb,
articleno = {5},
numpages = {32}
}

@inproceedings{10.5555/1283383.1283494,
author = {Arthur, David and Vassilvitskii, Sergei},
title = {k-means++: the advantages of careful seeding},
year = {2007},
isbn = {9780898716245},
publisher = {Society for Industrial and Applied Mathematics},
address = {USA},
booktitle = {Proceedings of the Eighteenth Annual ACM-SIAM Symposium on Discrete Algorithms},
pages = {1027–1035},
numpages = {9},
location = {New Orleans, Louisiana},
series = {SODA '07}
}

@article{10.1145/2395116.2395117,
author = {Ostrovsky, Rafail and Rabani, Yuval and Schulman, Leonard J. and Swamy, Chaitanya},
title = {The effectiveness of lloyd-type methods for the k-means problem},
year = {2013},
issue_date = {December 2012},
publisher = {Association for Computing Machinery},
address = {New York, NY, USA},
volume = {59},
number = {6},
issn = {0004-5411},
url = {https://doi.org/10.1145/2395116.2395117},
doi = {10.1145/2395116.2395117},
journal = {J. ACM},
month = jan,
articleno = {28},
numpages = {22}
}

@article{10.1145/3528087,
author = {Mitzenmacher, Michael and Vassilvitskii, Sergei},
title = {Algorithms with predictions},
year = {2022},
issue_date = {July 2022},
publisher = {Association for Computing Machinery},
address = {New York, NY, USA},
volume = {65},
number = {7},
issn = {0001-0782},
url = {https://doi.org/10.1145/3528087},
doi = {10.1145/3528087},
journal = {Commun. ACM},
month = jun,
pages = {33–35},
numpages = {3}
}

@inproceedings{10.5555/3495724.3496389,
author = {Antoniadis, Antonios and Gouleakis, Themis and Kleer, Pieter and Kolev, Pavel},
title = {Secretary and online matching problems with machine learned advice},
year = {2020},
isbn = {9781713829546},
publisher = {Curran Associates Inc.},
address = {Red Hook, NY, USA},
booktitle = {Proceedings of the 34th International Conference on Neural Information Processing Systems},
articleno = {665},
numpages = {12},
location = {Vancouver, BC, Canada},
series = {NIPS '20}
}

@inproceedings{Hsu2018LearningBasedFE,
  title={Learning-Based Frequency Estimation Algorithms},
  author={Chen-Yu Hsu and Piotr Indyk and Dina Katabi and A. Vakilian},
  booktitle={International Conference on Learning Representations},
  year={2018},

}

@inproceedings{10.5555/3540261.3541056,
author = {Dinitz, Michael and Im, Sungjin and Lavastida, Thomas and Moseley, Benjamin and Vassilvitskii, Sergei},
title = {Faster matchings via learned duals},
year = {2021},
isbn = {9781713845393},
publisher = {Curran Associates Inc.},
address = {Red Hook, NY, USA},
booktitle = {Proceedings of the 35th International Conference on Neural Information Processing Systems},
articleno = {795},
numpages = {14},
series = {NIPS '21}
}

@article{10.1145/3447579,
author = {Lykouris, Thodoris and Vassilvitskii, Sergei},
title = {Competitive Caching with Machine Learned Advice},
year = {2021},
issue_date = {August 2021},
publisher = {Association for Computing Machinery},
address = {New York, NY, USA},
volume = {68},
number = {4},
issn = {0004-5411},
url = {https://doi.org/10.1145/3447579},
doi = {10.1145/3447579},
journal = {J. ACM},
month = jul,
articleno = {24},
numpages = {25}
}

@InProceedings{pmlr-v178-gamlath22a,
  title = 	 {Approximate Cluster Recovery from Noisy Labels},
  author =       {Gamlath, Buddhima and Lattanzi, Silvio and Norouzi-Fard, Ashkan and Svensson, Ola},
  booktitle = 	 {Proceedings of Thirty Fifth Conference on Learning Theory},
  pages = 	 {1463--1509},
  year = 	 {2022},
  editor = 	 {Loh, Po-Ling and Raginsky, Maxim},
  volume = 	 {178},
  series = 	 {Proceedings of Machine Learning Research},
  month = 	 {02--05 Jul},
  publisher =    {PMLR},

  
}

@inproceedings{
Nguyen2022ImprovedLA,
title={Improved Learning-augmented Algorithms for k-means and k-medians Clustering},
author={Thy Dinh Nguyen and Anamay Chaturvedi and Huy Nguyen},
booktitle={The Eleventh International Conference on Learning Representations },
year={2023},
}

@inproceedings{huang2025new,
	title={New Algorithms for the Learning-Augmented k-means Problem},
	author={Junyu Huang and Qilong Feng and Ziyun Huang and Zhen Zhang and Jinhui Xu and Jianxin Wang},
	booktitle={The Thirteenth International Conference on Learning Representations},
	year={2025},

}

@inproceedings{10.1145/177424.178042,
author = {Inaba, Mary and Katoh, Naoki and Imai, Hiroshi},
title = {Applications of weighted Voronoi diagrams and randomization to variance-based k-clustering: (extended abstract)},
year = {1994},
isbn = {0897916484},
publisher = {Association for Computing Machinery},
address = {New York, NY, USA},
url = {https://doi.org/10.1145/177424.178042},
doi = {10.1145/177424.178042},
booktitle = {Proceedings of the Tenth Annual Symposium on Computational Geometry},
pages = {332–339},
numpages = {8},
location = {Stony Brook, New York, USA},
series = {SCG '94}
}

@Techreport{krizhevsky2009learning,
 author = {Krizhevsky, Alex and Hinton, Geoffrey},
 address = {Toronto, Ontario},
 institution = {University of Toronto},
 number = {0},
 publisher = {Technical report, University of Toronto},
 title = {Learning multiple layers of features from tiny images},
 year = {2009},
 title_with_no_special_chars = {Learning multiple layers of features from tiny images},
 url = {https://www.cs.toronto.edu/~kriz/learning-features-2009-TR.pdf}
}

@misc{kddcup2004,
  author       = {KDD},
  title        = {KDD Cup},
  howpublished = {\url{https://osmot.cs.cornell.edu/kddcup/index.html}},
  year         = {2004},

}

@ARTICLE{6296535,
  author={Deng, Li},
  journal={IEEE Signal Processing Magazine}, 
  title={The MNIST Database of Handwritten Digit Images for Machine Learning Research [Best of the Web]}, 
  year={2012},
  volume={29},
  number={6},
  pages={141-142},
  doi={10.1109/MSP.2012.2211477}}

@article{fashion_mnist,
  title={Fashion-MNIST: a Novel Image Dataset for Benchmarking Machine Learning Algorithms},
  author={Han Xiao and Kashif Rasul and Roland Vollgraf},
  journal={ArXiv},
  year={2017},
  volume={abs/1708.07747},
  url={https://api.semanticscholar.org/CorpusID:702279}
}

@book{Roughgarden_2021,author={Tim Roughgarden}, place={Cambridge}, title={Beyond the Worst-Case Analysis of Algorithms}, publisher={Cambridge University Press}, year={2021}}

\appendix
\section{missing proof for $k$-median}
\begin{claim}
\label{distanceinequalityforkmed}
The two distances $||\hat{c}_i-\mathtt{Med}(T_i)||_2$ and $||\hat{c}_i-c_i^*||_2$ are both no larger than $\frac{(2+\alpha\epsilon)\mathtt{Cost}(X_i^*, c_i^*)}{(1-2\alpha)|\tilde{\tilde{X_i}}|}$.
\end{claim}

\begin{proof}
   First, due to our greedy search approach, we can easily establish the inequality relationship between $\mathtt{Cost}(N_i(\hat{c}_i), \hat{c}_i)$ and $\mathtt{Cost}(T_{i}, \mathtt{Med}(T_{i}))$, we have
\begin{align}
    \mathtt{Cost}(N_i(\hat{c}_i), \hat{c}_i)\leqslant\mathtt{Cost}(N_i(q), q). \nonumber
\end{align}
As $N_i(q)$ is the set of the nearest points from $\hat{c}_i$  of size $ (1-\alpha)|\tilde{X_i}|$, we also have
\begin{align}
    \mathtt{Cost}(N_i(q), q)\leqslant\mathtt{Cost}(T_i, q). \nonumber
\end{align}
We have
\begin{align}
\mathtt{Cost}(T_i, q)\leqslant(1+\alpha\epsilon)\mathtt{Cost}(T_{i}, \mathtt{Med}(T_{i})).   \nonumber
\end{align}
From the inequalities presented above, it follows that
\begin{align}
\mathtt{Cost}(N_i(\hat{c}_i), \hat{c}_i)\leqslant(1+\alpha\epsilon)\mathtt{Cost}(T_{i}, \mathtt{Med}(T_{i})). \label{eq1}
\end{align}
   According to the definition of learing-augmented, $|T_i|\geqslant (1-\alpha)n_i$, and because $|N_i(\hat{c}_i)|=(1-\alpha)n_i$ and $N_i(\hat{c}_i), T_i\subseteq \tilde{X_i}$, We can derive 
\begin{align}
    |N_i(\hat{c}_i)\cap T_i|\geqslant|N_i(\hat{c}_i)|-|\tilde{X_i}\backslash T_i|&\geqslant(1-\alpha-\alpha)|\tilde{X_i}|\nonumber\\&=(1-2\alpha)|\tilde{X_i}|. \label{eq2}
\end{align}
Then, according to the triangle inequality and inequality \eqref{eq2} it follows that
\begin{align}
\label{eq3}
    &\sum_{p\in N_i(\hat{c}_i)\cap T_{i}}||p-\hat{c}_i||_{2}+||p-\mathtt{Med}(T_i)||_{2}\\\nonumber&\geqslant|N_i(\hat{c}_i)\cap T_i|\times||\hat{c}_i-T_i||_{2}\nonumber\\&\geqslant(1-2\alpha)|\tilde{X_i}|\times||\hat{c}_i-\mathtt{Med}(T_i)||_{2}.
\end{align}
Also by inequality \eqref{eq1}, we have 
\begin{align}
\sum_{p\in N_i(\hat{c}_i)\cap T_{i}}||p-\hat{c}_i||_{2}&\leqslant\mathtt{Cost}(N_i(\hat{c}_i), \hat{c}_i)\nonumber\\&\leqslant(1+\alpha\epsilon)\mathtt{Cost}(X_i^*, c_i^*)\label{eq20}
\end{align}
\begin{align}
\sum_{p\in N_i(\hat{c}_i)\cap T_{i}}||p-\mathtt{Med}(T_i)||_{2}\leqslant\mathtt{Cost}(X_i^*, c_i^*) \label{eq25}
\end{align}
Comprehensive inequalities \eqref{eq3}, \eqref{eq20}, and \eqref{eq25}, we obtain
\begin{align}
(2+\alpha\epsilon)\mathtt{Cost}(X_i^*, c_i^*)\geqslant(1-2\alpha)|\tilde{X_i}|||\hat{c}_i-\mathtt{Med}(T_i)||_{2}, 
\end{align}
which directly implies 
\begin{align}
||\hat{c}_i-\mathtt{Med}(T_i)||_2\leqslant \frac{(2+\alpha\epsilon)\mathtt{Cost}(X_i^*, c_i^*)}{(1-2\alpha)|\tilde{X_i}|}. 
\end{align}
Here we have completed the proof of the first inequality . 

The proof of the second inequality is similar, like inequality \eqref{eq3}, we have
\begin{align}
&\sum_{p\in N_i(\hat{c}_i)\cap X_{i}^*}||p-\hat{c}_i||_{2}+||p-\mathtt{Med}(X_{i}^*)||_{2}\nonumber\\&\geqslant|N_i(\hat{c}_i)\cap X_{i}^*|\times||\hat{c}_i-\mathtt{Med}(X_{i}^*)||_{2}\nonumber\\&\geqslant(1-2\alpha)|X_{i}|\times||\hat{c}_i-\mathtt{Med}(X_{i}^*)||_{2}, 
\end{align} 
the second inequality comes from $T_i\subseteq X_i^*$, so $|N_i(\hat{c}_i)\cap X_{i}^*|\geqslant|N_i(\hat{c}_i)\cap T_{i}^{*}|. $
Beacuse 
\begin{align}\sum_{p\in N_i(\hat{c}_i)\cap X_{i}^{*}}||p-\hat{c}_i||_{2}&\leqslant\mathtt{Cost}(N_i(\hat{c}_i), \hat{c}_i)\nonumber\\&\leqslant(1+\alpha\epsilon)\mathtt{Cost}(X_i^*, c_i^*)\nonumber
\end{align}
\begin{align}
&\sum_{p\in N_i(\hat{c}_i)\cap X_{i}^{*}}||p-\mathtt{Med}(X_{i}^{*})||_{2}\leqslant\mathtt{Cost}(X_i^*, c_i^*), \nonumber
\end{align}
we have
\begin{align}
||\hat{c}_i-c_i^*||_2\leqslant\frac{(2+\alpha\epsilon)\mathtt{Cost}(X_i^*, c_i^*)}{(1-2\alpha)|\tilde{X_i}|}. \nonumber
\end{align}
This is the second inequality we aimed to prove. 
\end{proof}

\section{Algorithm for $k$-means}
\textbf{Notations.} For k-means, given any point set C, the distance from a point p to its
closest point in C is denoted as $\mathtt{dist}_2(p,C)$. In particular, when C is the given set of centers for the point set X, the corresponding cost, denoted as $\mathtt{Cost}_2(X, C)$, is defined as $\mathtt{Cost}_2(X, C) = \sum_{x\in X} \mathtt{dist}_2(x, C).$ For any point set $P$ and we use $\mathtt{Cen}(P)$ to denote its \textbf{means point}, i.e.,
\begin{eqnarray}
\mathtt{Cen}(P)=\arg\min_{q\in \mathbb{R}^d}\sum_{p\in P}\|p-q\|^2_2.\nonumber
\end{eqnarray}

In this section, we extend the Sample-and-Search algorithm to Learning-augmented $k$-means problem. Our approach still proceeds in three stages, but differs from algorithm 1 in two aspects: we modify the number of sampled points and the method for building the candidate center set. Specifically, we first sample a constant-size set of data points, then construct candidate center sets in time exponential in the sample size, and finally identify locally optimal centers in time linear in the dataset size. 
\subsection{Our Proposed Algorithm And Main Theorem}
The detailed implementation of the algorithm is described in algorithm \ref{algorithm 2}. Table \ref{t1} provides a detailed comparison of results for Learning-Augmented $k$-means algorithms. 
\begin{table*}[htbp]
	\centering
	\caption{Comparison results of learning augmented $k$-means algorithms}
	\label{t1}
	\resizebox{\textwidth}{!}{%
		\begin{tabular}{lccc}
			\toprule
			Methods and References   & Approximation Ratio &  Label Error Range   & Time Complexity\\
			\midrule
			\citet{Ergun2021LearningAugmentedKC} & $1+20\alpha$ & $[\frac{10\log m}{\sqrt{m}}, 1/7]$ & $O(nd\log n)$\\
			\citet{Nguyen2022ImprovedLA} & $1+\frac{\alpha}{1-\alpha}+\frac{4\alpha}{(1-2\alpha)(1-\alpha)}$ & $[0, 1/2)$ & $O(nd\log n)$ \\
			\citet{huang2025new}  & $1+\frac{\alpha}{1-\alpha}+\frac{4\alpha+\alpha\epsilon}{(1-2\alpha)(1-\alpha)}$ & $[0, 1/2)$ & $O(\epsilon^{-1/2}nd\log(kd))$ \\
	    \citet{huang2025new} & $1+\frac{\alpha}{1-\alpha}+\frac{12\alpha-18\alpha^2}{(1-3\alpha-\epsilon)(1-2\alpha-\epsilon)}$ & $(0, 1/3-\epsilon)$ & $O(nd)+\tilde{O}(\epsilon^{-5}kd)$\\
			Sample-and-Search (ours) & $1+\frac{\alpha}{1-\alpha}+\frac{4\alpha+\alpha\epsilon}{(1-2\alpha)(1-\alpha)}$ & $[0, 1/2)$ & $O(2^{O(1/\epsilon)}nd\log k)$\\
			\bottomrule
		\end{tabular}%
	}
\end{table*}
\begin{algorithm}[htbp]
	\caption{Sample-and-Search for Learning-Augmented $k$-means}
	\label{algorithm 2}
        \begin{algorithmic}[1] 
	\STATE \textbf{Input:}A $k$-means instance $(X, k, d)$, a set $(\tilde{X_1}\ldots \tilde{X_k})$ of partitions with error rate $\alpha$, and a parameter $\epsilon\in (0, 1)$
	\STATE \textbf{Output:}A set $C \in \mathbb{R}^d $ of centers with $|C| = k$. 
	\STATE $\hat{C}\gets \{\}$\;
	\FOR{$i\in [k]$ }
		\STATE ${C}_i\gets\{\}$\;
		\FOR{$j=1$ to $\lceil\frac{\log(\delta/k)}{\log 0.75}\rceil$}
		\STATE Randomly and independently sample a set ${R}^j_i$ from $\tilde{X_i}$ with size $\lceil\frac{4}{(1-\alpha)\epsilon}\rceil$\;
			\FOR { every $\lceil\frac{1}{(1-\alpha)\epsilon}\rceil$ subset ${R}$ of ${R}^j_i$ }
            \STATE ${C}_i={C}_i\cup \mathtt{Cen}_2(R)$
                \ENDFOR
		\ENDFOR
		\STATE For each $c \in {C}_i$, define $N_i(c)$ as the set of $\lceil(1-\alpha)m_i\rceil$ points in $\tilde{X_i}$ closest to $c$.
		\STATE $\hat{c}_i=\arg\min_{c\in C^{'}}\mathtt{Cost}_2(N_i(c), c)$\;
		\STATE $\hat{C}=\hat{C}\cup \mathtt{Cen}(N_i(\hat{c}_i))$
		\ENDFOR
	\STATE \textbf{Return} $\hat{C}$
    \end{algorithmic}
\end{algorithm}
We present the main theoretical result of our algorithm
below.
\begin{theorem}
\label{ag2}
   Algorithm \ref{algorithm 2} is an algorithm for Learning-Augmented $k$-means clustering. Given a dataset $X \in \mathbb{R}^{n \times d}$ and a partition $(X_1, \ldots, X_k)$ with an error rate of $\alpha < \frac{1}{2}$, the algorithm outputs a solution with an approximation ratio of
$1+\frac{\alpha}{1-\alpha}+\frac{4\alpha+\alpha\epsilon}{(1-2\alpha)(1-\alpha)}$
and completes with constant probability within $O(2^{O(1/\varepsilon)} nd \log k)$ time complexity. 
\end{theorem}

\subsection{Proof of Theorem \ref{ag2}}

We first introduce two well-known and widely used results in the field of k-means clustering. 
\begin{proposition} 
\label{traforkmea}
Let $X \subseteq \mathbb{R}^d$ be a set of $n$ points, and $c \in \mathbb{R}^d$. Then:
\begin{align}
    \textup{Cost}_2(X, c) = \textup{Cost}_2(X, \textup{Cen}_2(X)) + n \cdot \|c - \textup{Cen}_2(X)\|_2^2. \nonumber
\end{align}
\end{proposition}

\begin{proposition} 
\label{divide}
For any arbitrary partition $X_1 \cup X_2$ of a set $X \subseteq \mathbb{R}^d$, where $X$ has size $n$, if $|X_1| \ge (1-\lambda)n$, then:
\begin{align}
    \|\mathtt{Cen}_2(X), \mathtt{Cen}_2(X_1)\|_2^2 \le \frac{\lambda}{(1-\lambda)n} \mathtt{Cost}_2(X, \mathtt{Cen}_2(X)). \nonumber
\end{align}
\end{proposition}
We also introduce a important propositions on geometric means point in Euclidean space, which are essential for our following proofs.
\begin{proposition}\cite{10.1145/177424.178042}
\label{lmforkme}
Let $S$ be a set of $m$ points obtained by independently sampling $m$ points uniformly at random from a point set $P$. Then, for any $\delta > 0$, 
\begin{align}
   ||\textup{Cen}_2(P)-\textup{Cen}_2(T_i)||_2 \leqslant \frac{ \epsilon \mathtt{Cost}_2(T_i, c(T_i))}{ |T_i|}. \nonumber
\end{align}\nonumber
holds with probability at least $1-\delta$. 
\end{proposition}
Proposition \ref{lmforkme} shows that if a sufficient number of points are sampled randomly from $P$, then the centroid of the sampled points is close to $\mathtt{Med}(P)$ with high probability, which enables us to construct a candidate set of centers by  directly using the centroid of the sampled points.

We divide the proof into three main steps: First, we establish that for each predicted cluster, with high probability, the constructed candidate center set contains at least one point that is close to the true median of the correctly labeled subset of the predicted cluster. This is formalized in Lemma \ref{lmecandidateset}, which leverages the geometric properties from Proposition \ref{lmforkme} under our sampling design. Second, we analyze the cost of the selected center from the candidate set. In Lemma \ref{thapproximationratioforkmea}, we show that this center yields a clustering cost close to the optimal one, de-
spite the noisy labels, by carefully bounding the additional
cost incurred by misclassified points and the optimality of
the greedy choice. Finally, we aggregate the bounds over all
clusters to obtain the total clustering cost, and analyze the
size of the candidate set and runtime of our algorithm.

\begin{lemma}
\label{lmecandidateset}
For each pridicted cluster $\tilde{X_i}$, with probability $1-\frac{\delta}{k}$, there exists a point $q \in C'$ satisfying:
\begin{align}
||q-\mathtt{Cen}_2(T_i)||_2 \leqslant \frac{ \epsilon \mathtt{Cost}_2(T_i, c(T_i))}{ |T_i|}. \nonumber
\end{align}
\end{lemma}
\begin{proof}
First, under the learning-augmented setting, we have 
\begin{align}
    |T_i|\geqslant(1-\alpha)\max(|\tilde{X_i}|, |\tilde{X_i}|). \nonumber
\end{align} Then, as we sampled a point set ${Q}_i$ of size $\frac{4}{(1-\alpha)\epsilon}$ from $\tilde{X_i}$, by employing Markov’s inequality, we deduce that,
\begin{align}
|T_i\cap {R}^j_i|\geqslant \frac{2}{\epsilon}. \nonumber
\end{align}
with probability at least $\frac{1}{2}$.
It follows that there exist a subset $R$ satisfies
\begin{align}
    R\subset T_i ,|R|=\frac{2}{\epsilon}. \nonumber
\end{align}

By  Proposition \ref{lmforkme}, let $ P=T_i, S=R ,\delta=\frac{1}{2},m=\frac{2}{\epsilon}$, we have
\begin{align}
||q-\mathtt{Cen}_2(T_i)||_2 \leqslant \frac{ \epsilon \mathtt{Cost}_2(T_i, c(T_i))}{ |T_i|}.  \nonumber
\end{align}
weth probability at least $\frac{1}{2}$.
Now, we calculate the probability that all events succeed in a single trial. The combined success probability is $\frac{1}{4}$, Since we performl $\lceil\frac{\log(\delta/k)}{\log 0.75}\rceil$ runs, the overall success probability is therefore
greater than 
$$1-0.75^{\frac{\log(\delta/k)}{\log 0.75}}=1-\frac{\delta}{k}.$$
\end{proof}
We now turn to evaluate the clustering cost incurred by
the selected centers. Lemma \ref{thapproximationratioforkmea} plays a central role in this
analysis, it quantifies how far the selected center might be
from the true cluster center due to noisy labels and sampling
variance, and how this error translates into overall clustering
cost. We first establish the following claim.
\begin{claim} 
\label{disineforkmea}
The distance  between $\hat{c}_i$ and $\textup{Cen}_2(T_i)$  satisfies the following inequality:
    \begin{align}
		||\hat{c}_i-\textup{Cen}_2(T_i)||_2^2\leq\frac{4\alpha+\epsilon\alpha}{1-2\alpha}\frac{\operatorname{Cost}_2(T_{i}, \textup{Cen}_2(T_i ))}{|T_i|}. \nonumber
	\end{align}
\end{claim}
\begin{proof}
First, due to our greedy search process, we have
\begin{align}
\mathtt{Cost}_2(N_i(\hat{c}_i), \hat{c}_i)\leqslant \mathtt{Cost}_2(N_i(q), q). \label{eq9}
\end{align}
As $N_i(q)$ is the set of the nearest points from $\hat{c}_i$  of size $ (1-\alpha)|X_i|$, we also have
\begin{align}
    \mathtt{Cost}_2(N_i(q), q)\leqslant\mathtt{Cost}_2(T_i, q). \label{eq50}
\end{align}
Lemma \ref{lmecandidateset} directly yields
\begin{align}
\mathtt{Cost}_2(T_i, q)\leqslant(1+\epsilon)\mathtt{Cost}_2(T_{i}, \mathtt{Cen}_2(T_{i})). \label{eq55}
\end{align}
Since $|T_i|\geqslant(1-\alpha)n_i$, from the inequalities \eqref{eq9}, \eqref{eq50} and \eqref{eq55} presented above, it follows that
\begin{align}
    \frac{\mathtt{Cost}_2(N_i(\hat{c}_i), \hat{c}_i)}{(1-\alpha)n_i}\leqslant (1+\epsilon)\frac{\mathtt{Cost}_2(T_i, \mathtt{Cen}_2(T_i))}{|T_i|}. \nonumber
\end{align}
So, we have
\begin{align}
   &\mathtt{Cost}_2(N_i(\hat{c}_i), \hat{c}_i)\nonumber\\&\leqslant(1+\epsilon)\frac{(1-\alpha)n_i }{|T_i|}\mathtt{Cost}_2(T_i, \mathtt{Cen}_2(T_i))), \label{eq60}
\end{align}
According to Proposition  \ref{divide}, if we set $C=N_i(\hat{c}_i), C_1=N_i(\hat{c}_i)\cap T_i$, as we know
\begin{align}
    \frac{|C_1|}{C}\geqslant \frac{|N_i(\hat{c}_i)|-|X_i\backslash T_i|  }{|N_i(\hat{c}_i)|}\geqslant \frac {1-2\alpha}     {1-\alpha}  =1-\frac{\alpha} {1-\alpha}\nonumber, 
\end{align}
then it follows that
\begin{align}
\label{eq13}
    ||\hat{c}_i-\mathtt{Cen}_2(N_i(\hat{c}_i)\cap T_i)||_2^2&\leqslant 
\frac {\frac{\alpha} {1-\alpha}} {1-\frac{\alpha} {1-\alpha} } \frac{\mathtt{Cost}_2(N_i(\hat{c}_i), \hat{c}_i)}{|N_i(\hat{c}_i)|}\nonumber\\&=\frac{\alpha  \mathtt{Cost}_2( N_i(\hat{c}_i), \hat{c}_i)}{(1-2\alpha)(1-\alpha)n_i}. 
\end{align}
Similarly, if we set $C=T_i, C_1=N_i(\hat{c}_i)\cap T_i$, we can also bound the distance between $\mathtt{Cen}_2(N_i(\hat{c}_i)\cap T_i)$ and $\mathtt{Cen}_2(T_i)$ as
\begin{align}
\label{eq500}
||\mathtt{Cen}_2(T_i)-\mathtt{Cen}_2(N_i(\hat{c}_i)\cap T_i)||_2^2\leqslant \frac{\alpha  \text{Cost}(T_i, \mathtt{Cen}_2(T_i)}{(1-2\alpha)|T_i|}, 
\end{align}
Based on inequalities \eqref{eq60}, \eqref{eq13} and \eqref{eq500}, we are able to bound the distance between $\hat{c}_i$ and $\mathtt{Cen}_2(T_i)$
\begin{align}
||\hat{c}_i-\mathtt{Cen}_2(T_i)||_2^{2}&\leqslant2||{\hat{c}_i}-\mathtt{Cen}_2(N_i(\hat{c}_i)\cap T_i)||_2^{2}\nonumber\\&\quad+2||\mathtt{Cen}_2(N_i(\hat{c}_i)\cap T_i)-\mathtt{Cen}_2(T_i)||_2^{2}\nonumber\\
& \leqslant 2\frac{\alpha  \text{Cost}_2( N_i(\hat{c}_i), \hat{c}_i)}{(1-2\alpha)(1-\alpha)n_i}\nonumber\\&\quad+2\frac{\alpha  \text{Cost}_2( T_i, \mathtt{Cen}_2(T_i)}{(1-2\alpha)|T_i|}\nonumber\\
& =\frac{(4\alpha+2\epsilon\alpha) \text{Cost}_2( T_i, \mathtt{Cen}_2(T_i))}{(1-2\alpha)|T_i|}.\nonumber 
\end{align}

\end{proof}
We now turn to evaluate the clustering cost incurred by the selected centers. Lemma \ref{thapproximationratioforkmea} plays a central role in this analysis—it quantifies how far the selected center might be from the true cluster center due to noisy labels and sampling variance, and how this error translates into overall clustering cost.
\begin{lemma}
\label{thapproximationratioforkmea}
For pridicted cluster $X_i$, we have:
\begin{align}
  &\mathtt{Cost}(X_i^*, \hat{c}_i)\nonumber\\&\leqslant\left(1+\frac{\alpha}{1-\alpha}+\frac{4\alpha+\alpha\epsilon}{(1-2\alpha)(1-\alpha)}\right)\mathtt{Cost}(X_i^*, c_i^*). \nonumber
\end{align}
\end{lemma}
\begin{proof}
 We begin by dividing the calculation of $\mathtt{Cost}_2(X_i^*,c_i^*) $ into two parts
\begin{align}
\label{eq75}
\text{Cost}_2(X_i^{*}, c_i^*)&=\text{Cost}_2(X_i^{*}\backslash T_i, c_i^*)+\text{Cost}_2(T_i, c_i^*). \nonumber\\
\end{align}
According to Proposition \ref{traforkmea}, $\mathtt{Cost}_2(T_i, c_i^*)$ can be written as
\begin{align}
\label{eq16}
&\text{Cost}_2(T_i, \mathtt{Cen}_2(T_i))\nonumber\\&\quad+(1-\frac{|X_i^{*}\backslash T_i|}{|X_i^{*}|})|{X_i^*}|\times||c_i^*-\mathtt{Cen}_2(T_i)||_2^{2}. 
\end{align}
Similarly, we have
\begin{align}
\label{eq17}
\text{Cost}_2(X_i^{*}/ T_i, c_i^*)&=\text{Cost}_2(X_i^{*}\backslash T_i, \mathtt{Cen}_2(X_i^*\backslash T_i))\nonumber\\&+\frac{|X_i^{*}\backslash T_i|}{|X_i^{*}|}|X_i^{*}|\times||c_i^*-\mathtt{Cen}_2(X_i^*\backslash T_i)||_2^{2}
\end{align}
Combining inequalities \ref{eq75}, \ref{eq16}, and \ref{eq17}
We obtain
\begin{align}
\text{Cost}_2(X_i^{*}, c_i^*)&=\text{Cost}_2(X_i^{*}\backslash T_i, \mathtt{Cen}_2(X_i^*\backslash T_i))\nonumber\\&\quad+\frac{|X_i^{*}\backslash T_i|}{|X_i^{*}|}|X_i^{*}|||c_i^*-\mathtt{Cen}_2(X_i^*/ T_i)||_2^{2}\nonumber\\&\quad+\text{Cost}_2(T_i, \mathtt{Cen}_2(T_i))\nonumber\\&\quad+(1-\frac{|X_i^{*}\backslash T_i|}{|X_i^{*}|})|X_i^{*}|||c_i^*-\mathtt{Cen}_2(T_i)||_2^{2}\nonumber\\&=\frac{1-\frac{|X_i^{*}\backslash T_i|}{|X_i^{*}|}}{\frac{|X_i^{*}\backslash T_i|}{|X_i^{*}|}}|X_i^{*}|||c_i^*-\mathtt{Cen}_2(T_i)||_2^{2}\nonumber\\&\quad+\text{Cost}_2(X_i^{*}/ T_i, \mathtt{Cen}_2(X_i^*/ T_i))\nonumber\\&\quad+\text{Cost}_2(T_i, \mathtt{Cen}_2(T_i))\nonumber\\&\geqslant\frac{1-\alpha}{\alpha}|X_i^{*}|||c_i^*-\mathtt{Cen}_2(T_i)||_2^{2}\nonumber\\&\quad+\text{Cost}_2(T_i, \mathtt{Cen}_2(T_i))\nonumber\\&\geqslant\frac{1-\alpha}{\alpha}|X_i^{*}|||c_i^*-\mathtt{Cen}_2(T_i)||_2^{2}\nonumber\nonumber\\&\quad+\frac{1-2\alpha}{4\alpha+2\epsilon\alpha}\cdot(1-\alpha)|\tilde{X_i^*}|||\hat{c}_i-\mathtt{Cen}_2(T_{i})||_2^{2}, 
\end{align}
 
 Finally, by the Cauchy-Schwarz inequality, we obtain
\begin{align}&\left(||c_i^*-\mathtt{Cen}_2(T_i)||_2+||\hat{c}_i-\mathtt{Cen}_2(T_i)||_2\right)^{2}\nonumber\\&\leq\left(\frac{\alpha}{1-\alpha}+\frac{4\alpha+2\epsilon\alpha}{(1-2\alpha)(1-\alpha)}\right)\nonumber\\&\quad\times(\frac{1-\alpha}{\alpha}|X_i^{*}|\times||c_i^*-\mathtt{Cen}_2(T_i)||_2^{2}\nonumber\\&\quad+\frac{1-2\alpha}{4\alpha+2\epsilon\alpha}\cdot(1-\alpha)|X_i^{*}|\times||\hat{c}_i-\mathtt{Cen}_2(T_i)||_2^{2})/|X_i^{*}|\nonumber\\&\leq(\frac{\alpha}{1-\alpha}+\frac{4\alpha+2\epsilon\alpha}{(1-2\alpha)(1-\alpha)})\mathtt{Cost}_2(X_i^{*}, c_i^*)/|X_i^{*}|. \nonumber
\end{align}
This directly yields 
\begin{align}
    &||\hat{c}_i-\mathtt{Cen}_2(X_i^*||_2^2\nonumber\\&\leqslant\mathtt{Cost}_2(X_i^{*}, c_i^*)\left(\frac{\alpha}{1-\alpha}+\frac{4\alpha+2\epsilon\alpha}{(1-2\alpha)(1-\alpha)}\right)/|X_i^{*}|. \nonumber
\end{align}
By Proposition \ref{traforkmea}, this is equivalent to what we aim to prove. 
\end{proof}
We now proceed to formally prove Theorem 2.1 by establishing both the approximation ratio and the runtime complexity.

\begin{proof}[Proof of \Cref{ag2}]
   Our first step is to compute the approximation ratio of the algorithm. In each cluster, by Lemma \ref{thapproximationratioforkmea}, we obtain 
\begin{align}
    \mathtt{Cost}_2(X_i^*, \hat{c}_i)
	\leqslant \left(1+\frac{\alpha}{1-\alpha}+\frac{4\alpha+\alpha\epsilon}{(1-2\alpha)(1-\alpha)}\right)\mathtt{Cost}_2(X_i^*, c_i^*). 
\end{align}
Therefore, for the instance, we have
\begin{align}&\sum_{i\in[k]}\mathtt{Cost}_2(X_{i}^{*}, \{\widehat{c}_{j}\}_{j=1}^{k})	\nonumber\\&\leqslant \left(1+\frac{\alpha}{1-\alpha}+\frac{4\alpha+\alpha\epsilon}{(1-2\alpha)(1-\alpha)}\right)\sum_{i\in[k]}\mathtt{Cost}_2(X_i^*, c_i^*). 
\end{align}
Next, we analyze the time complexity of the algorithm. 
First, we compute the size of set of candidate centers. 
The total count of subsets of S with a fixed size is
$\tbinom{O(\frac{1}{\varepsilon(1-\alpha)})}{1/\epsilon}\leqslant2^{O(1/\epsilon)}$, 
The time required to construct the candidate set is
$O\left(2^{O(1/\epsilon)}d\right)$. For each candidate point within the candidate center set, the time needed to calculate its cost is $n_id$. Consequently, the overall time complexity of the algorithm is
\begin{align}
\sum_{i\in[k]}2^{O(1/\varepsilon)}|\tilde{X_i}|d\log \frac{k}{\delta}=2^{O(1/\varepsilon)} nd\log \frac{k}{\delta}. 
\end{align}
Then, we analyze the success probability of the algorithm. 
By Lemma \ref{lmecandidateset} the success probability within a single cluster is $1-\frac{k}{\delta}.$
By the union bound, the overall success probability of the algorithm $\geqslant 1-k\times\frac{k}{\delta}=1-\delta$. 
\end{proof} 
\newpage
\section{Additional experiment for Learning-Augment $k$-median}
 Tables 5-8 show the experimental results on datasets CIFAR-10, Fashion-Mnist, PHY and Mnist,  for varying $\alpha$ with fixed $k$. Table 9-11 shows our results 
for varying $k$ with fixed $\alpha$. Both sets of results demonstrate that our algorithm is substantially faster than all competing methods while generally achieving better approximation quality ,particularly on high-dimensional datasets. 
\begin{table*}[htbp]
\centering
\caption{Performance comparison on Cifar10 dataset with fixed $k=10$ and varied $\alpha$.}
\label{tab:cifar10_varied_alpha}
\footnotesize
\setlength{\tabcolsep}{1.5mm}
\begin{tabular}{
    l
    l
    c c
    c c
    c c
    c c
}
\toprule
\multicolumn{2}{c}{Condition} & \multicolumn{2}{c}{Cost} & \multicolumn{2}{c}{Time(s)} & \multicolumn{2}{c}{NMI} & \multicolumn{2}{c}{ARI} \\
\cmidrule(lr){3-4} \cmidrule(lr){5-6} \cmidrule(lr){7-8} \cmidrule(lr){9-10}
{$\alpha$} & {Method} & {Avg.} & {Std. Dev.} & {Avg.} & {Std. Dev.} & {Avg.} & {Std. Dev.} & {Avg.} & {Std. Dev.} \\
\midrule
{0} & K-Med++ & 2.7492e+07 & - & - & - & - & - & - & - \\
\midrule

\multirow{5}{*}{0.1}
& Predictor & 2.7528e+07 & - & - & - & - & - & - & - \\
& EFSplus      & 2.7406e+07 & 1109.14 & 29.87 & 14.76 & \textbf{0.6631} & 0.0012 & \textbf{0.5918} & 0.0026 \\
& HFHplus    & 2.7401e+07 & 3929.91 & 184.22 & 1.45  & 0.6562 & 0.0078 & 0.5862 & 0.0107 \\
& NCN      & 2.7404e+07 & 3466.13 & 34.43  & 7.95  & 0.6588 & 0.0050 & 0.5862 & 0.0065 \\
& Ours     & \textbf{2.7397e+07} & 5089.44 & \textbf{14.16}  & 0.26  & 0.6538 & 0.0090 & 0.5815 & 0.0131 \\
\bottomrule
\end{tabular}
\end{table*}

\begin{table*}[htbp]
\centering
\caption{Performance comparison on Fashion-MNIST dataset with fixed $k=10$ and varied $\alpha$.}
\label{tab:fashion_mnist_varied_alpha_full}
\footnotesize
\setlength{\tabcolsep}{1.5mm}
\begin{tabular}{
    l
    l
    c c
    c c
    c c
    c c
}
\toprule
\multicolumn{2}{c}{Condition} & \multicolumn{2}{c}{Cost} & \multicolumn{2}{c}{Time(s)} & \multicolumn{2}{c}{NMI} & \multicolumn{2}{c}{ARI} \\
\cmidrule(lr){3-4} \cmidrule(lr){5-6} \cmidrule(lr){7-8} \cmidrule(lr){9-10}
{$\alpha$} & {Method} & {Avg.} & {Std. Dev.} & {Avg.} & {Std. Dev.} & {Avg.} & {Std. Dev.} & {Avg.} & {Std. Dev.} \\
\midrule
{0} & K-Med++ & 8.4054e+07 & - & - & - & - & - & - & - \\
\midrule

\multirow{5}{*}{0.1}
& Predictor & 8.4259e+07 & - & - & - & - & - & - & - \\
& EFSplus      & 8.4050e+07 & 115.17 & 270.47 & 12.97 & \textbf{0.9558} & 0.0003 & \textbf{0.9642} & 0.0003 \\
& HFHplus    & 8.4049e+07 & 834.92 & 749.72 & 18.47 & 0.9528 & 0.0036 & 0.9591 & 0.0050 \\
& NCN      & 8.4050e+07 & 181.80 & 272.22 & 4.37  & 0.9552 & 0.0008 & 0.9636 & 0.0007 \\
& Ours     & \textbf{8.4048e+07} & 933.64 & \textbf{47.37}  & 0.78  & 0.9516 & 0.0019 & 0.9581 & 0.0022 \\
\bottomrule
\end{tabular}
\end{table*}

\begin{table*}[htbp]
\centering
\caption{Performance comparison on PHY dataset with fixed $k=10$ and varied $\alpha$.}
\label{tab:phy_varied_alpha}
\footnotesize
\setlength{\tabcolsep}{1.5mm}
\begin{tabular}{
    l
    l
    c c
    c c
    c c
    c c
}
\toprule
\multicolumn{2}{c}{Condition} & \multicolumn{2}{c}{Cost} & \multicolumn{2}{c}{Time(s)} & \multicolumn{2}{c}{NMI} & \multicolumn{2}{c}{ARI} \\
\cmidrule(lr){3-4} \cmidrule(lr){5-6} \cmidrule(lr){7-8} \cmidrule(lr){9-10}
{$\alpha$} & {Method} & {Avg.} & {Std. Dev.} & {Avg.} & {Std. Dev.} & {Avg.} & {Std. Dev.} & {Avg.} & {Std. Dev.} \\
\midrule
{0} & K-Med++ & 1.9797e+08 & - & - & - & - & - & - & - \\
\midrule

\multirow{5}{*}{0.1}
& Predictor & 1.9849e+08 & - & - & - & - & - & - & - \\
& EFSplus      & 1.9800e+08 & \textbf{478.11} & 392.68 & 99.57 & 0.9882 & 0.0001 & 0.9900 & 0.0001 \\
& HFHplus    & 1.9799e+08 & 6778.35 & 39.54  & 6.32  & \textbf{0.9896} & 0.0017 & \textbf{0.9913} & 0.0017 \\
& NCN      & 1.9800e+08 & 4363.70 & 376.60 & 102.88 & 0.9883 & 0.0005 & 0.9900 & 0.0005 \\
& Ours     & \textbf{1.9798e+08} & 7427.30 & \textbf{34.05}  & 2.86  & 0.9867 & 0.0049 & 0.9884 & 0.0053 \\
\bottomrule
\end{tabular}
\end{table*}
\begin{table*}[htbp]
\centering
\caption{Performance comparison on MNIST dataset with fixed $k=10$ and varied $\alpha$.}
\label{tab:mnist_varied_alpha}
\footnotesize
\setlength{\tabcolsep}{1.5mm}
\begin{tabular}{
    l
    l
    c c
    c c
    c c
    c c
}
\toprule
\multicolumn{2}{c}{Condition} & \multicolumn{2}{c}{Cost} & \multicolumn{2}{c}{Time(s)} & \multicolumn{2}{c}{NMI} & \multicolumn{2}{c}{ARI} \\
\cmidrule(lr){3-4} \cmidrule(lr){5-6} \cmidrule(lr){7-8} \cmidrule(lr){9-10}
{$\alpha$} & {Method} & {Avg.} & {Std. Dev.} & {Avg.} & {Std. Dev.} & {Avg.} & {Std. Dev.} & {Avg.} & {Std. Dev.} \\
\midrule
{0} & K-Med++ & 44696.31 & - & - & - & - & - & - & - \\
\midrule

\multirow{5}{*}{0.1}
& Predictor & 44795.92 & - & - & - & - & - & - & - \\
& EFSplus      & 44695.47 & 0.78 & 25.23 & 6.42 & 0.9721 & 0.0030 & 0.9722 & 0.0038 \\
& HFHplus    & 44695.66 & 0.40 & 3.58  & 1.09 & 0.9682 & 0.0069 & 0.9688 & 0.0071 \\
& NCN      & 44695.55 & 0.47 & 24.97 & 6.58 & \textbf{0.9739} & 0.0053 & \textbf{0.9746} & 0.0056 \\
& Ours     & \textbf{44695.31} & 0.71 & \textbf{3.36}  & 0.91  & 0.9688 & 0.0072 & 0.9692 & 0.0072 \\
\bottomrule
\end{tabular}
\end{table*}

\begin{table*}[htbp]
\centering
\caption{Performance comparison on Fashion-MNIST dataset with fixed $\alpha=0.2$ and varied $k$.}
\label{tab:fashion_mnist_varied_k}
\footnotesize
\renewcommand{\arraystretch}{0.9}
\setlength{\tabcolsep}{1.5mm}
\begin{tabular}{
    l
    l
    c c
    c c
    c c
    c c
}
\toprule
\multicolumn{2}{c}{Condition} & \multicolumn{2}{c}{Cost} & \multicolumn{2}{c}{Time(s)} & \multicolumn{2}{c}{NMI} & \multicolumn{2}{c}{ARI} \\
\cmidrule(lr){3-4} \cmidrule(lr){5-6} \cmidrule(lr){7-8} \cmidrule(lr){9-10}
{k} & {Method} & {Avg.} & {Std. Dev.} & {Avg.} & {Std. Dev.} & {Avg.} & {Std. Dev.} & {Avg.} & {Std. Dev.} \\
\midrule

\multirow{6}{*}{10}
& K-Med++   & 8.4058e+07 & - & - & - & - & - & - & - \\
& Predictor & 8.4938e+07 & - & - & - & - & - & - & - \\
& EFSplus       & 8.4059e+07 & 305.75 & 281.82 & 11.34 & 0.9460 & 0.0005 & 0.9547 & 0.0006 \\
& HFHplus     & 8.4053e+07 & 1129.16 & 733.73 & 24.48 & 0.9419 & 0.0035 & 0.9489 & 0.0057 \\
& NCN       & 8.4059e+07 & 811.26 & 285.90 & 2.63  & 0.9455 & 0.0015 & 0.9541 & 0.0016 \\
& Ours      & \textbf{8.4052e+07} & 1275.81 & \textbf{60.78}  & 3.40  & \textbf{0.9473} & 0.0025 & \textbf{0.9557} & 0.0030 \\
\bottomrule
\end{tabular}
\end{table*}

\begin{table*}[htbp]
\centering
\caption{Performance comparison on PHY dataset with fixed $\alpha=0.2$ and varied $k$.}
\label{tab:phy_varied_k_full}
\footnotesize
\setlength{\tabcolsep}{1.5mm}
\begin{tabular}{
    l
    l
    c c
    c c
    c c
    c c
}
\toprule
\multicolumn{2}{c}{Condition} & \multicolumn{2}{c}{Cost} & \multicolumn{2}{c}{Time(s)} & \multicolumn{2}{c}{NMI} & \multicolumn{2}{c}{ARI} \\
\cmidrule(lr){3-4} \cmidrule(lr){5-6} \cmidrule(lr){7-8} \cmidrule(lr){9-10}
{k} & {Method} & {Avg.} & {Std. Dev.} & {Avg.} & {Std. Dev.} & {Avg.} & {Std. Dev.} & {Avg.} & {Std. Dev.} \\
\midrule

\multirow{6}{*}{10}
& K-Med++   & 2.0224e+08 & - & - & - & - & - & - & - \\
& Predictor & 2.0428e+08 & - & - & - & - & - & - & - \\
& EFSplus       & 2.0205e+08 & 4444.80 & 362.31 & 52.30 & \textbf{0.9369} & 0.0002 & \textbf{0.9192} & 0.0003 \\
& HFHplus     & 2.0147e+08 & 105661.96 & 42.33  & 1.91  & 0.9289 & 0.0139 & 0.9036 & 0.0240 \\
& NCN       & 2.0163e+08 & 109131.95 & 160.15 & 14.33 & 0.9173 & 0.0097 & 0.8866 & 0.0150 \\
& Ours      & \textbf{2.0135e+08} & 82812.18 & \textbf{20.72}  & 0.59  & 0.8934 & 0.0222 & 0.8446 & 0.0360 \\
\bottomrule
\end{tabular}
\end{table*}

\begin{table*}[htbp]
\centering
\caption{Performance comparison on MNIST dataset with fixed $\alpha=0.2$ and varied $k$.}
\label{tab:mnist_varied_k}
\footnotesize
\setlength{\tabcolsep}{1.5mm}
\begin{tabular}{
    l
    l
    c c
    c c
    c c
    c c
}
\toprule
\multicolumn{2}{c}{Condition} & \multicolumn{2}{c}{Cost} & \multicolumn{2}{c}{Time(s)} & \multicolumn{2}{c}{NMI} & \multicolumn{2}{c}{ARI} \\
\cmidrule(lr){3-4} \cmidrule(lr){5-6} \cmidrule(lr){7-8} \cmidrule(lr){9-10}
{k} & {Method} & {Avg.} & {Std. Dev.} & {Avg.} & {Std. Dev.} & {Avg.} & {Std. Dev.} & {Avg.} & {Std. Dev.} \\
\midrule

\multirow{6}{*}{10}
& K-Med++           & 44695.78 & - & - & - & 1.0000 & - & 1.0000 & - \\
& Predictor         & 45101.25 & - & - & - & - & - & - & - \\
& NCN               & 44700.26 & 0.66 & 4.68  & 0.36 & 0.9738 & 0.0031 & 0.9768 & 0.0026 \\
& Fast-Estimation   & 44698.33 & 0.78 & 2.18  & 0.15 & 0.9655 & 0.0037 & 0.9670 & 0.0044 \\
& Ngu               & 44701.73 & 0.97 & 4.53  & 0.44 & \textbf{0.9745} & 0.0016 & \textbf{0.9774} & 0.0013 \\
& Ours              & \textbf{44697.49} & 0.58 & \textbf{0.28}  & 0.03 & 0.9659 & 0.0028 & 0.9677 & 0.0024 \\
\bottomrule
\end{tabular}
\end{table*}

\section{Experiment for Learning-Augment $k$-means}
We evaluated our algorithms on real-world datasets. The experiments were conducted on a server with an Intel(R) Xeon(R) Gold 6154 CPU and 1024GB of RAM. For all experiments, we report the average clustering cost and its standard deviation over 10 independent runs.

\textbf{Datasets}. Following the work of \citet{Nguyen2022ImprovedLA}, \citet{Ergun2021LearningAugmentedKC} and \citet{huang2025new}, we evaluate our algorithms on the CIFAR-10 $(n=50,000,d=3,072)$ \cite{krizhevsky2009learning}, PHY $(n=10,000,d=50)$ \cite{kddcup2004}, and MNIST  $(n=1,797,d=64)$ \cite{6296535} datasets using a range of error rates $\alpha$. We additionally evaluated our algorithm's performance on the Fashion-MNIST $(n=60000, d=784)$ \cite{fashion_mnist} dataset to assess its efficacy in high-dimensional datasets."

\textbf{Predictor Generation and Error Simulation} To evaluate our algorithms, we first computed a ground-truth partition for each dataset using Lloyd’s algorithm initialized
with KMeans++(denoteed as KMe++). We then generated corrupted partitions with the error rate, $\alpha$, by randomly selecting an $\alpha$ fraction of points in each true cluster and reassigning them to randomly chosen cluster(denoted as Predictor). To ensure a fair comparison, every algorithm was tested on the exact same set of corrupted labels for any given error rate $\alpha$.

\textbf{Algorithms} In our experiments, we evaluate our proposed Sample-and-Search algorithm . We compare their performance against other state-of-the-art learning-augmented methods, including the algorithm from \citet{Ergun2021LearningAugmentedKC}(denoted as Erg),\citet{Nguyen2022ImprovedLA}(denoted as Ngu) and the recent work by \citet{huang2025new}(denoted as Fast-Sampling,Fast-Estimation and
Fast-Filtering). We utilize the cost calculated from the labels generated by our predictor as the baseline. The cost computed from the undegraded labels is considered the optimal cost. Our primary comparison focuses on the clustering cost and algorithm runtime on the given dataset. Furthermore, we also computed the standard deviation to assess the stability of the algorithms

\textbf{Results.} Table 12-15 show the experimental results on datasets CIFAR-10, Fashion-Mnist, PHY and Mnist,  for varying $\alpha$ with fixed $k$. Table 16-18 shows our results for varying $k$ with fixed $\alpha$. Both sets of results demonstrate that our algorithm is a bit slower than Fast-Filtering method  while generally achieving better approximation quality ,particularly on high-dimensional datasets. 

\begin{table*}[htbp]
\centering
\caption{Performance comparison on CIFAR-10 dataset with fixed $k=10$ and varied $\alpha$.}
\label{tab:cifar10_varied_alpha_means}
\footnotesize
\setlength{\tabcolsep}{1.5mm}
\begin{tabular}{
    l
    l
    c c
    c c
    c c
    c c
}
\toprule
\multicolumn{2}{c}{Condition} & \multicolumn{2}{c}{Cost} & \multicolumn{2}{c}{Time(s)} & \multicolumn{2}{c}{NMI} & \multicolumn{2}{c}{ARI} \\
\cmidrule(lr){3-4} \cmidrule(lr){5-6} \cmidrule(lr){7-8} \cmidrule(lr){9-10}
{$\alpha$} & {Method} & {Avg.} & {Std. Dev.} & {Avg.} & {Std. Dev.} & {Avg.} & {Std. Dev.} & {Avg.} & {Std. Dev.} \\
\midrule
{0} & KMe++ & 3.9481e+11 & - & - & - & - & - & - & - \\
\midrule

\multirow{7}{*}{0.1}
& Predictor         & 3.9680e+11 & - & - & - & - & - & - & - \\
& EFSplus           & 4.2240e+11 & 6.5569e+07 & 7.86 & 0.22 & 0.8930 & 0.0004 & 0.8853 & 0.0005 \\
& NCN               & 3.9869e+11 & 0 & 127.44 & 6.20 & 0.8980 & 0.0000 & 0.8943 & 0.0000 \\
& Fast-Estimation   & 3.9747e+11 & 5.5532e+08 & 7.62 & 0.29 & 0.8777 & 0.0093 & 0.8681 & 0.0122 \\
& Fast-Filtering    & 3.9680e+11 & 2.5731e+08 & \textbf{0.73} & 0.07 & 0.9201 & 0.0064 & 0.9245 & 0.0073 \\
\bottomrule
\end{tabular}
\end{table*}
\begin{table*}[htbp]
\centering
\caption{Performance comparison on Fashion-MNIST dataset with fixed $k=10$ and varied $\alpha$.}
\label{tab:fashion_mnist_varied_alpha_full_means}
\footnotesize
\setlength{\tabcolsep}{1.5mm}
\begin{tabular}{
    l
    l
    c c
    c c
    c c
    c c
}
\toprule
\multicolumn{2}{c}{Condition} & \multicolumn{2}{c}{Cost} & \multicolumn{2}{c}{Time(s)} & \multicolumn{2}{c}{NMI} & \multicolumn{2}{c}{ARI} \\
\cmidrule(lr){3-4} \cmidrule(lr){5-6} \cmidrule(lr){7-8} \cmidrule(lr){9-10}
{$\alpha$} & {Method} & {Avg.} & {Std. Dev.} & {Avg.} & {Std. Dev.} & {Avg.} & {Std. Dev.} & {Avg.} & {Std. Dev.} \\
\midrule
{0} & K-Med++ & 1.2450e+11 & - & - & - & - & - & - & - \\
\midrule

\multirow{7}{*}{0.1}
& Predictor         & 1.2613e+11 & - & - & - & - & - & - & - \\
& EFSplus           & 1.2669e+11 & 7.7786e+06 & 2.917 & 0.556 & 0.8991 & 0.0002 & 0.8898 & 0.0004 \\
& NCN               & 1.2567e+11 & 0.00 & 3.367 & 0.814 & 0.9510 & 0.0000 & 0.9570 & 0.0000 \\
& Fast-Estimation   & 1.2575e+11 & 8.3179e+07 & 2.637 & 0.791 & 0.9325 & 0.0115 & 0.9339 & 0.0142 \\
& Fast-Filtering    & 1.2481e+11 & 3.0078e+07 & \textbf{0.731} & 0.319 & 0.9588 & 0.0046 & 0.9650 & 0.0056 \\
\bottomrule
\end{tabular}
\end{table*}
\begin{table*}[htbp]
\centering
\caption{Performance comparison on MNIST dataset with fixed $k=10$ and varied $\alpha$.}
\label{tab:mnist_varied_alpha_means}
\footnotesize
\setlength{\tabcolsep}{1.5mm}
\begin{tabular}{
    l
    l
    c c
    c c
    c c
    c c
}
\toprule
\multicolumn{2}{c}{Condition} & \multicolumn{2}{c}{Cost} & \multicolumn{2}{c}{Time(s)} & \multicolumn{2}{c}{NMI} & \multicolumn{2}{c}{ARI} \\
\cmidrule(lr){3-4} \cmidrule(lr){5-6} \cmidrule(lr){7-8} \cmidrule(lr){9-10}
{$\alpha$} & {Method} & {Avg.} & {Std. Dev.} & {Avg.} & {Std. Dev.} & {Avg.} & {Std. Dev.} & {Avg.} & {Std. Dev.} \\
\midrule
{0} & KMe++ & 1.1653e+06 & - & - & - & - & - & - & - \\
\midrule

\multirow{7}{*}{0.1}
& Predictor         & 1.1790e+06 & - & - & - & - & - & - & - \\
& EFSplus           & 1.1867e+06 & 1057.86 & 0.080 & 0.050 & 0.9687 & 0.0034 & 0.9714 & 0.0033 \\
& NCN               & 1.1747e+06 & 0.00 & 0.057 & 0.035 & 0.9778 & 0.0000 & 0.9801 & 0.0000 \\
& Fast-Estimation   & 1.1757e+06 & 315.60 & 0.087 & 0.037 & 0.9770 & 0.0017 & 0.9793 & 0.0021 \\
& Fast-Filtering    & 1.1683e+06 & 251.37 & \textbf{0.013} & 0.013 & 0.9775 & 0.0043 & 0.9803 & 0.0043 \\
\bottomrule
\end{tabular}
\end{table*}
\begin{table*}[htbp]
\centering
\caption{Performance comparison on PHY dataset with fixed $k=10$ and varied $\alpha$.}
\label{tab:phy_varied_alpha_means}
\footnotesize
\setlength{\tabcolsep}{1.5mm}
\begin{tabular}{
    l
    l
    c c
    c c
    c c
    c c
}
\toprule
\multicolumn{2}{c}{Condition} & \multicolumn{2}{c}{Cost} & \multicolumn{2}{c}{Time(s)} & \multicolumn{2}{c}{NMI} & \multicolumn{2}{c}{ARI} \\
\cmidrule(lr){3-4} \cmidrule(lr){5-6} \cmidrule(lr){7-8} \cmidrule(lr){9-10}
{$\alpha$} & {Method} & {Avg.} & {Std. Dev.} & {Avg.} & {Std. Dev.} & {Avg.} & {Std. Dev.} & {Avg.} & {Std. Dev.} \\
\midrule
{0} & K-Med++ & 1.0147e+12 & - & - & - & - & - & - & - \\
\midrule

\multirow{7}{*}{0.1}
& Predictor         & 1.1078e+12 & - & - & - & - & - & - & - \\
& EFSplus           & 1.0230e+12 & 6.5840e+08 & 0.357 & 0.138 & 0.9769 & 0.0052 & 0.9758 & 0.0066 \\
& NCN               & 1.0167e+12 & 0.00 & 1.083 & 0.274 & 0.9937 & 0.0000 & 0.9953 & 0.0000 \\
& Fast-Estimation   & 1.0239e+12 & 3.2252e+09 & 0.284 & 0.090 & 0.9776 & 0.0063 & 0.9778 & 0.0079 \\
& Fast-Filtering    & 1.0171e+12 & 5.2853e+08 & \textbf{0.091} & 0.050 & 0.9913 & 0.0050 & 0.9924 & 0.0060 \\
\bottomrule
\end{tabular}
\end{table*}

\begin{table*}[htbp]
\centering
\caption{Performance comparison on CIFAR-10 dataset with fixed $\alpha=0.2$ and varied $k$.}
\label{tab:cifar10_varied_k_means}
\footnotesize
\setlength{\tabcolsep}{1.5mm}
\begin{tabular}{
    l
    l
    c c
    c c
    c c
    c c
}
\toprule
\multicolumn{2}{c}{Condition} & \multicolumn{2}{c}{Cost} & \multicolumn{2}{c}{Time(s)} & \multicolumn{2}{c}{NMI} & \multicolumn{2}{c}{ARI} \\
\cmidrule(lr){3-4} \cmidrule(lr){5-6} \cmidrule(lr){7-8} \cmidrule(lr){9-10}
{k} & {Method} & {Avg.} & {Std. Dev.} & {Avg.} & {Std. Dev.} & {Avg.} & {Std. Dev.} & {Avg.} & {Std. Dev.} \\
\midrule

\multirow{8}{*}{10}
& K-Med++           & 7.8584e+10 & - & - & - & - & - & - & - \\
& Predictor         & 8.0164e+10 & - & - & - & - & - & - & - \\
& EFSplus           & 7.9181e+10 & 2.64e+06 & 1.347 & 0.025 & 0.8625 & 0.0008 & 0.8403 & 0.0008 \\
& NCN               & 7.9408e+10 & 0.00 & 3.678 & 0.030 & 0.8282 & 0.0000 & 0.7948 & 0.0000 \\
& Fast-Estimation   & 7.9677e+10 & 2.54e+08 & 1.345 & 0.022 & 0.7905 & 0.0183 & 0.7350 & 0.0259 \\
& Fast-Filtering    & 7.8971e+10 & 3.99e+07 & 0.163 & 0.005 & 0.8786 & 0.0039 & 0.8673 & 0.0052 \\
\bottomrule
\end{tabular}
\end{table*}
\begin{table*}[htbp]
\centering
\caption{Performance comparison on Fashion-MNIST dataset with fixed $\alpha=0.2$ and varied $k$.}
\label{tab:fashion_mnist_varied_k_new}
\footnotesize
\setlength{\tabcolsep}{1.5mm}
\begin{tabular}{
    l
    l
    c c
    c c
    c c
    c c
}
\toprule
\multicolumn{2}{c}{Condition} & \multicolumn{2}{c}{Cost} & \multicolumn{2}{c}{Time(s)} & \multicolumn{2}{c}{NMI} & \multicolumn{2}{c}{ARI} \\
\cmidrule(lr){3-4} \cmidrule(lr){5-6} \cmidrule(lr){7-8} \cmidrule(lr){9-10}
{k} & {Method} & {Avg.} & {Std. Dev.} & {Avg.} & {Std. Dev.} & {Avg.} & {Std. Dev.} & {Avg.} & {Std. Dev.} \\
\midrule

\multirow{8}{*}{10}
& K-Med++           & 1.2454e+11 & - & - & - & - & - & - & - \\
& Predictor         & 1.3036e+11 & - & - & - & - & - & - & - \\
& EFSplus           & 1.3029e+11 & 1.07e+07 & 1.227 & 0.004 & 0.8876 & 0.0005 & 0.8857 & 0.0006 \\
& NCN               & 1.2941e+11 & 0.00 & 1.446 & 0.002 & 0.9215 & 0.0000 & 0.9266 & 0.0000 \\
& Fast-Estimation   & 1.2933e+11 & 6.37e+08 & 1.203 & 0.003 & 0.9091 & 0.0214 & 0.9105 & 0.0267 \\
& Fast-Filtering    & 1.2576e+11 & 1.35e+08 & 0.272 & 0.006 & 0.9279 & 0.0148 & 0.9338 & 0.0211 \\
\bottomrule
\end{tabular}
\end{table*}

\begin{table*}[htbp]
\centering
\caption{Performance comparison on PHY dataset with fixed $\alpha=0.2$ and varied $k$.}
\label{tab:phy_varied_k_means}
\footnotesize
\setlength{\tabcolsep}{1.5mm}
\begin{tabular}{
    l
    l
    c c
    c c
    c c
    c c
}
\toprule
\multicolumn{2}{c}{Condition} & \multicolumn{2}{c}{Cost} & \multicolumn{2}{c}{Time(s)} & \multicolumn{2}{c}{NMI} & \multicolumn{2}{c}{ARI} \\
\cmidrule(lr){3-4} \cmidrule(lr){5-6} \cmidrule(lr){7-8} \cmidrule(lr){9-10}
{k} & {Method} & {Avg.} & {Std. Dev.} & {Avg.} & {Std. Dev.} & {Avg.} & {Std. Dev.} & {Avg.} & {Std. Dev.} \\
\midrule

\multirow{8}{*}{10}
& K-Med++           & 1.0147e+12 & - & - & - & - & - & - & - \\
& Predictor         & 1.3965e+12 & - & - & - & - & - & - & - \\
& EFSplus           & 1.0562e+12 & 4.32e+09 & 0.109 & 0.000 & 0.9591 & 0.0045 & 0.9523 & 0.0068 \\
& NCN               & 1.0209e+12 & 0.00 & 0.411 & 0.001 & 0.9829 & 0.0000 & 0.9834 & 0.0000 \\
& Fast-Estimation   & 1.0657e+12 & 1.48e+10 & 0.096 & 0.000 & 0.9530 & 0.0149 & 0.9439 & 0.0218 \\
& Fast-Filtering    & 1.0246e+12 & 1.43e+08 & \textbf{0.020} & 0.000 & \textbf{0.9882} & 0.0007 & \textbf{0.9902} & 0.0008 \\
\bottomrule
\end{tabular}
\end{table*}

\end{document}